\documentclass[a4paper,10pt]{article}
\usepackage{jcappub} % for details on the use of the package, please see the JINST-author-manual
\usepackage{lineno}
\usepackage{tensor}
\usepackage{subcaption}
\usepackage{tikz}
\usepackage{multirow}
\usepackage{amsmath}
\usepackage{amsthm}
\usetikzlibrary{positioning, shapes.geometric}
\newtheorem{theorem}{Theorem}[section]

\newtheorem{corollary}[theorem]{Corollary}

\newtheorem{definition}[theorem]{Definition}

\newtheorem{proposition}[theorem]{Proposition}
\newtheorem{remark}[theorem]{Remark}

\arxivnumber{2505.15975} % Only if you have one
\title{\boldmath Geometric formulation of \texorpdfstring{$k$}{k}-essence and late-time acceleration}

\author[a,b]{Lehel Csillag,}

\affiliation[a]{Department of Physics, Babeș-Bolyai University, Kogălniceanu Street, Cluj-Napoca 400084, Romania}
\affiliation[b]{Faculty of Mathematics and Computer Science, Transilvania University, Iuliu Maniu Street 50, Brașov 500091, Romania}
\author[c]{Erik Jensko}
\affiliation[c]{Department of Mathematics, University College London,
Gower Street, London WC1E 6BT, UK
}%

\emailAdd{lehel.csillag@ubbcluj.ro, lehel.csillag@unitbv.ro, erik.jensko@ucl.ac.uk}

\date{\today}

\abstract{ 
We study a class of geometries in which nonmetricity is fully determined by a vectorial degree of freedom and three independent coefficients. Formulating the simplest linear action in this geometry, implemented through Lagrange multipliers, naturally leads to an equivalence with the purely kinetic $k$-essence models with quadratic kinetic terms. A detailed dynamical systems analysis reveals that the $\Lambda$CDM phenomenology is embedded within the model.
Crucially, we find that if stability conditions such as a positive sound speed squared and energy density  are not enforced, the model generically exhibits instabilities and divergent behaviour in the phase space. These physical viability criteria allow us to isolate stable regions of the parameter space and derive well-motivated priors for parameter inference. Using Markov Chain Monte Carlo methods and late-time observational data, including cosmic chronometers, Pantheon$^{+}$ Type Ia supernovae, and DESI
baryon acoustic oscillations, we constrain the degrees of freedom associated with nonmetricity and demonstrate the viability of the model.
We discuss the implications of these results in light of the recent cosmic tensions, and give a possible explanation as to why the equivalent $k$-essence models have been missed as serious competitors to $\Lambda$CDM in the past. Finally, we review the geometric foundations of the theory and show that the integrable Weyl, Schr\"{o}dinger and completely symmetric geometries are embedded within our framework as special cases.
}

\begin{document}
\maketitle
\flushbottom
\section{Introduction}

The standard $\Lambda$-cold-dark-matter ($\Lambda$CDM) cosmological model has proven to be exceptionally good at describing observational phenomena on a range of scales, 
from the observed accelerated expansion of the Universe ~\cite{SupernovaSearchTeam:1998fmf} to the formation of large-scale structure \cite{SDSS:2005xqv}.
This is exemplified in the detailed cosmological surveys of the past decade, with the predictions of $\Lambda$CDM being confirmed by a range of observational probes \cite{Planck:2018vyg}.
However, despite the successes of the phenomenological model, our present understanding of dark matter and dark energy is incomplete \cite{Joyce:2014kja,Peebles:2002gy}. 
Moreover, the cosmological constant faces theoretical problems \cite{Weinberg:1988cp}, which potentially motivate dynamical explanations for dark energy \cite{Ratra:1987rm,Caldwell:1997ii,Copeland:2006wr}. Some of the simplest examples are scalar field models,
including the well-known quintessence scenario \cite{Wetterich:1994bg} and the more general $k$-essence models \cite{Armendariz-Picon:1999hyi,Armendariz-Picon:2000ulo} (originally proposed in the context of inflation but shortly after studied in relation to dark energy \cite{Armendariz-Picon:2000nqq,Chiba:1999ka}). Such theories are natural candidates to study alternative cosmological models to $\Lambda$CDM. 

On the observational side, current tensions in cosmology may also call for extensions to $\Lambda$CDM and modifications of General Relativity (GR) \cite{Abdalla:2022yfr,DiValentino:2025sru}. One prominent example is the long-standing \textit{Hubble tension}: the discrepancy between the value of the Hubble constant $H_0$ inferred from early-time observations ($67.4 \pm 0.5$~km/s/Mpc from Planck 2018 \cite{Planck:2018vyg}) and from late-time observations ($73.04 \pm 1.04$~km/s/Mpc from the SH0ES collaboration \cite{Riess:2021jrx}).
 This tension has been shown to be statistically significant and difficult to ease within the standard $\Lambda$CDM framework \cite{Verde:2019ivm,DiValentino:2020zio}. Moreover, it does not appear to be lessening with the onset of new observational data; for instance, the most recent surveys of the Dark Energy Spectroscopic Instrument (DESI) further point towards tensions in $\Lambda$CDM, with dynamical dark energy being preferred by the data \cite{DESI:2025zpo,DESI:2025zgx}. In fact, even simple scalar field models have been shown to perform much better than $\Lambda$CDM in this regard \cite{Keeley:2025rlg,Gialamas:2025pwv}.
  It is therefore worthwhile to explore and test alternative models.

In light of these theoretical issues and observational challenges, modified theories of gravity have become a popular approach from which to construct new cosmological models, see \cite{Goenner:2004se,Goenner:2014mka,Clifton:2011jh,Nojiri:2010wj} for reviews.
One particular class of theories with a long and rich history, overlapping significantly with particle physics and stemming from a gauge-theoretic approach to gravitation, is the metric-affine theory \cite{Blagojevic:2013xpa,Hehl:1994ue}.
Well-known examples include the Einstein-Cartan \cite{Kibble:1961ba,Hehl:1976kj,Trautman:2006fp,Shapiro:2001rz,Luz:2023uhy}, Einstein-Weyl \cite{Weyl:1918ib,calderbank1997einstein,Wheeler:2018rjb}, and teleparallel theories of gravity \cite{Hayashi:1967se,Nester:1998mp,Obukhov:2002tm,Aldrovandi:2013wha}, though many more possibilities exist \cite{Hehl:1994ue,Puetzfeld:2004yg}. These frameworks have since been used to construct a myriad of modified gravity theories \cite{Sotiriou:2008rp,Bengochea:2008gz,Linder:2010py,Cai:2015emx,BeltranJimenez:2017tkd,BeltranJimenez:2019esp,BeltranJimenez:2019tme,Bahamonde:2021gfp,Boehmer:2021aji,Boehmer:2023fyl,Jensko:2023lmn,Csillag:2024eor,Csillag:2024oqo,Bamba:2012cp,Nojiri:2017ncd,Gialamas:2024iyu}. Of particular interest are the extended Weyl geometries introduced in \cite{aringazin1991matter}, possessing both curvature and vectorial nonmetricity. This was later generalised to connections with linear vectorial distortion \cite{BeltranJimenez:2015pnp,BeltranJimenez:2016wxw}, which was shown to give rise to promising inflationary scenarios. However, the compatibility of these types of models with cosmological observations has yet to be studied in detail; this motivates the present study, in which late-time observational data is used to constrain the new degrees of freedom associated with nonmetricity. 

In this work, we further extend the Weyl geometric framework by allowing for specific vectorial nonmetricity contributions, motivated below. Besides having a natural geometric interpretation, the benefits of these generalisations become immediately obvious when studying cosmological applications,
where the new terms display phenomenologically appealing properties. By working with integrable vectorial nonmetricity, defined solely in terms of the metric and scalar fields, we present our model in a scalar-tensor representation. Remarkably, this reveals a direct equivalence\footnote{Similar dualities can be found for the Palatini $R^2$ theories \cite{Guendelman:2015jii,Gialamas:2019nly} and the $k$-essence effective geometries \cite{Panda:2024amk,Ganguly:2025kdf}, which we refer to for comparison. Moreover, effective geometry in $k$-essence is related to the Weyl integrable nonmetricity \cite{Sawicki:2024ryt}. For more general equivalences between non-canonical scalar-tensor models and geometric modifications of gravity, see \cite{Yang:2024kdo}.} with the $k$-essence theories.  Our specific model, constructed from the Ricci scalar alone, corresponds to a quadratic form of  \textit{purely kinetic} $k$-essence \cite{Scherrer:2004au}, also known as a ghost condensate model \cite{Arkani-Hamed:2003pdi}.
This model has one additional parameter compared to $\Lambda$CDM, and we derive constraints on the free model parameter based on stability arguments (a positive sound speed squared $c_s^2 \geq 0$ and the absence of ghosts). Such models have seen particular interest due to their potential ability to describe a unified dark sector \cite{Scherrer:2004au,Bertacca:2007ux,Bertacca:2010ct,Guendelman:2015jii}. Our explicit analytic solutions confirm this behaviour at late times, with the $k$-essence fluid evolving like a cosmological constant at zeroth order and dark matter at first order.

Interestingly, until now, this particular model had not been compared with observational data.
In fact, recent results have suggested that the quadratic $k$-essence model is not viable for describing late-time cosmology altogether \cite{Quiros:2025fnn}. On the contrary, our results directly show that late-time $\Lambda$CDM dynamics are embedded within the model. Using a combination of cosmic chronometers (CC), Type Ia supernovae (SNe Ia) and Baryon Acoustic  Oscillation (BAO) data from the Pantheon$^{+}$ and DESI DR1 samples, the model is shown to be compatible with late-time observations.
The cosmological implications are further investigated with a detailed dynamical systems analysis, supporting our observational results and clarifying some points of contention in the recent literature. From the geometric perspective, we thus conclude that integrable vectorial nonmetricity is a feasible candidate to describe dark energy. 

In Section \ref{sec:2} we review the metric-affine formalism and the foundations of nonmetricity. We introduce our new form of integrable vectorial nonmetricity and discuss how this relates to previous studies. In Section \ref{sec:3} we formulate our action principle and derive the equations of motion. The scalar-tensor representation is then presented, which naturally lends itself to a $k$-essence description. Utilising this language, we discuss stability constraints related to the $k$-essence fluid energy density and sound speed squared. 
Focusing on cosmological applications, we first conduct a detailed dynamical systems analysis of the model in Section \ref{sec:4}. This makes the $\Lambda$CDM limit and the general viability of the cosmological scenario immediately clear. In Section \ref{sec:5} we perform an observational analysis using SNe Ia, CC and BAO data. We show that due to the inherent instabilities and unphysical regions of parameter space, precise knowledge of the parameter priors is needed for obtaining correct solutions. These are found from the preceding stability analyses, after which we perform Markov Chain Monte Carlo methods to constrain the model parameters. The reduced chi-squared statistic, Akaike information criterion and Bayesian information criterion are then used to compare the geometric $k$-essence model with $\Lambda$CDM.
%This is paired with the $k$-essence perturbative stability constraints to give further information about the parameter priors. 
We conclude in Section \ref{sec:6} with a discussion of the theoretical framework, observational implications and directions for future research.

\section{Geometric foundations of nonmetricity} \label{sec:2}
Given a pseudo-Riemannian manifold\footnote{We work with a Lorentzian metric with signature $(-,+,+,+)$.} $(M,g)$, a covariant derivative $\nabla$, locally described by its connection coefficient functions $\tensor{\Gamma}{^\mu _\nu_\rho}$, is not prescribed a priori. This raises the question of how this choice of choosing a connection relates to the underlying metric. It turns out that this freedom is completely determined by two geometric objects \cite{schouten2013ricci}
\begin{equation}
    \tensor{T}{^\mu _\nu _\rho}:=\tensor{\Gamma}{^\mu _\rho _\nu} - \tensor{\Gamma}{^\mu _\nu _\rho} \, , \qquad Q_{\mu \nu \rho}:=-\nabla_{\mu} g_{\nu \rho} \,, 
\end{equation}
where
\begin{itemize}
    \item[-] $\tensor{T}{^\mu _\nu _\rho}$ is the torsion tensor, measuring how much the connection fails to be symmetric, i.e. the extent to which the parallelogram law fails to hold infinitesimally,
    \item[-] $Q_{\mu \nu \rho}$ is the nonmetricity tensor, measuring how the length of a vector changes during parallel transport.
\end{itemize}

The nonmetricity tensor is often decomposed into a trace and trace-free part \cite{Hehl:1994ue}
\begin{equation} \label{Q_decomp}
    Q_{\mu \nu \rho} = Q_{\mu} g_{\nu \rho} + \overline{Q}{}_{\mu \nu \rho} \, ,
\end{equation}
where $Q_{\mu}=\tensor{Q}{_\mu _\lambda ^\lambda}/4$ is called the  \textit{dilation} vector, and $\overline{Q}_{\mu \lambda}{}^{\lambda}=0$ defines the trace-free part known as the \textit{shear} \cite{Hehl:1976kt,Hehl:1976kv}. The dilation is the well-known Weyl vector term \cite{Weyl:1918ib} and preserves relative lengths and angles under parallel transport. Notably, it follows that the light-cone structure is also kept intact. The shear does not preserve either relative lengths or angles, but volumes remain unchanged. When coupled to matter, the decomposition (\ref{Q_decomp}) gives rise to associated dilation and shear hypermomentum currents \cite{Neeman:1996zcr,Hehl:1994ue}. Although we do not explicitly include hypermomentum couplings in this work, we later discuss the dilation and shear contributions corresponding to geometries with a specific form of nonmetricity.

An arbitrary affine connection can be decomposed in a post-Riemannian sense into the Levi-Civita connection and additional non-Riemannian components as
\begin{equation}
    \tensor{{\Gamma}}{^\mu _\nu _\rho}=\overset{\circ}{\Gamma}\tensor{}{^\mu _\nu _\rho}+\tensor{N}{^\mu _\rho _\nu} \, , 
\end{equation}
where
\begin{itemize}
    \item[-] $\overset{\circ}{\Gamma} \tensor{}{^\mu _\nu _\rho}$ are the Christoffel symbols of the Levi-Civita connection,
    \item[-] $\tensor{N}{^\mu _\rho _\nu }$ is the distortion tensor, accounting for torsion and nonmetricity effects \cite{schouten2013ricci}.
\end{itemize}
 The distortion tensor $\tensor{N}{^\mu _\rho _\nu }$ is explicitly given by
\begin{equation} \label{distortion}
    \tensor{N}{^\mu _\rho _\nu }= \frac{1}{2} g^{\lambda \mu}(-Q_{\lambda \nu \rho}+ Q_{\rho \lambda \nu} + Q_{\nu \rho \lambda}) - \frac{1}{2}g^{\lambda \mu}(T_{\rho \nu \lambda}+T_{\nu \rho \lambda}- T{}_{\lambda \rho \nu}) \, .
\end{equation}
We adopt the convention where the Riemann tensor is defined by
\begin{equation} 
\tensor{R}{^\mu _\nu _\rho _\sigma}=\tensor{\Gamma}{^\lambda _\nu _\sigma} \tensor{\Gamma}{^\mu _\lambda _\rho}-\tensor{\Gamma}{^\lambda _\nu _\rho}\tensor{\Gamma}{^\mu _\lambda _\sigma} +\partial_{\rho} \tensor{\Gamma}{^\mu _\nu _\sigma}- \partial_{\sigma} \tensor{\Gamma}{^\mu _\nu _\rho} \,.
\end{equation}
Using the decomposition of the connection, the Riemann tensor can be expressed as
\begin{equation} 
    \tensor{R}{^\mu _\nu _\rho _\sigma}=\overset{\circ}{R} \tensor{}{^\mu _\nu _\rho_\sigma}+\overset{\circ}{\nabla}_{\rho} \tensor{N}{^\mu _\sigma _\nu}- \overset{\circ}{\nabla}_{\sigma} \tensor{N}{^\mu_\rho _\nu}+\tensor{N}{^{\mu} _\rho_\lambda}\tensor{N}{^\lambda _\sigma _\nu}- \tensor{N}{^\mu _\sigma _\lambda} \tensor{N}{^\lambda _\rho _\nu} \, ,
\end{equation}
where $\overset{\circ}{R} \tensor{}{^\mu _\nu _\rho _\sigma}$ is the Riemann tensor of the Levi-Civita connection.
The Ricci tensor and Ricci scalar are then obtained by contraction,
\begin{align}
R_{\nu \sigma}&=\tensor{R}{^\mu _\nu _\mu _\sigma}=\overset{\circ}{R}_{\nu \sigma}+\overset{\circ}{\nabla}_{\mu} \tensor{N}{^\mu _\sigma _\nu} - \overset{\circ}{\nabla}_{\sigma} \tensor{N}{^\mu _\mu _\nu} + \tensor{N}{^\mu _\mu _\beta}\tensor{N}{^\beta _\sigma _\nu} - \tensor{N}{^\mu _\sigma _\beta}\tensor{N}{^\beta _\mu _\nu}\,, \\
    R&=g^{\nu \sigma} R_{\nu \sigma}=\overset{\circ}{R}+\overset{\circ}{\nabla}_{\mu} \tensor{N}{^\mu ^\nu _\nu} - \overset{\circ}{\nabla}_{\sigma} \tensor{N}{^\mu _\mu ^\sigma} + \tensor{N}{^ \mu  _\mu _\beta} \tensor{N}{^\beta ^\nu _\nu} - \tensor{N}{^\mu _\sigma _\beta} \tensor{N}{^\beta _\mu ^\sigma}\,. \label{Ricci_decomp}
\end{align}
In this paper, we consider a symmetric affine connection with a  novel form of vectorial nonmetricity given by\footnote{Surprisingly, to the best of our knowledge, this form of nonmetricity has rarely appeared before in the literature \cite{Iosifidis:2018jwu,Iosifidis:2020gth,Andrei:2024vvy}. The corresponding geometric quantities, such as distortion, dilation, shear and curvature tensors, have not been studied before for this nonmetricity.}
\begin{equation}\label{vectorial_nonmetricity}
    Q_{\mu \nu \rho}= c_1 \pi_{\mu} g_{\nu \rho} + c_2\left(\pi_{\rho} g_{\mu \nu} +\pi_{\nu} g_{\rho \mu} \right)+2c_3 \pi_{\mu} \pi_{\nu} \pi_{\rho} \, ,
\end{equation}
where $\pi_{\mu}$ are the components of a one-form and $c_1$, $c_2$, $c_3$ are constants. The specific form of this type of nonmetricity can be motivated by comparing with previous works on cosmological hypermomentum, where a completely analogous $Q_{\mu \nu \lambda}$ appears on FLRW backgrounds \cite{Iosifidis:2020gth,Andrei:2024vvy}. Moreover, the physical role of these $c_2$ and $c_3$ terms will become clear when we study the cosmologies of these models.

Additionally, if the components of the one-form $\pi_{\mu}$ are  locally the gradient of a smooth function, i.e. $\pi_{\mu}=\partial_{\mu} \phi$, the nonmetricity is said to be \textit{integrable}. The first $c_1$ term appears in the standard Weyl geometries \cite{Weyl:1918ib}, while the $c_2$ terms can be found in the previously discussed works on vectorial distortion \cite{aringazin1991matter,BeltranJimenez:2015pnp}. Notably, the completely symmetric $c_3$ has rarely been studied in physics, with its mathematical interpretation discussed in detail in Appendix \ref{appendix_coofree}. 
%It also fits nicely with the recent works on cosmological hypermomentum, where a completely analogous term appears on FLRW backgrounds \cite{Iosifidis:2020gth,Andrei:2024vvy}.

The geometric effects induced by the nonmetricity \eqref{vectorial_nonmetricity} are fully determined by the coefficients $c_1,c_2,c_3$, their magnitudes and signs, and the structure of the one-form $\pi$. For specific choices of these coefficients, one can recover the Weyl \cite{Weyl:1918ib}, Schrödinger \cite{schrodinger1985space,Csillag:2024eor} and completely symmetric geometries \cite{Csillag_2024statmfd}. These subcases are illustrated in Fig.~\ref{fig:vectorial_nonmetricity}. Although these are special cases of \eqref{vectorial_nonmetricity}, their geometric properties differ significantly. For a thorough coordinate-free description, we refer the reader to Appendix \ref{appendix_coofree}. Here, we summarise their main properties:
\begin{itemize}
    \item[$\triangleright$] \textit{Weyl geometry:} volumes and vector lengths are not preserved under parallel transport, but angles are. If the nonmetricity is integrable, volumes are preserved as well.
    \item[$\triangleright$] \textit{Schrödinger geometry:} certain special vectors, namely the autoparallels, preserve their length, while angles and volumes generally change. Assuming integrability again restores volume preservation.
    \item[$\triangleright$] \textit{Completely symmetric geometry:} in the most general case, parallel transport does not preserve any geometric quantities -- angles, lengths, or volumes. However, when $c_3=0$ the connection admits volume-preserving transport.
\end{itemize}
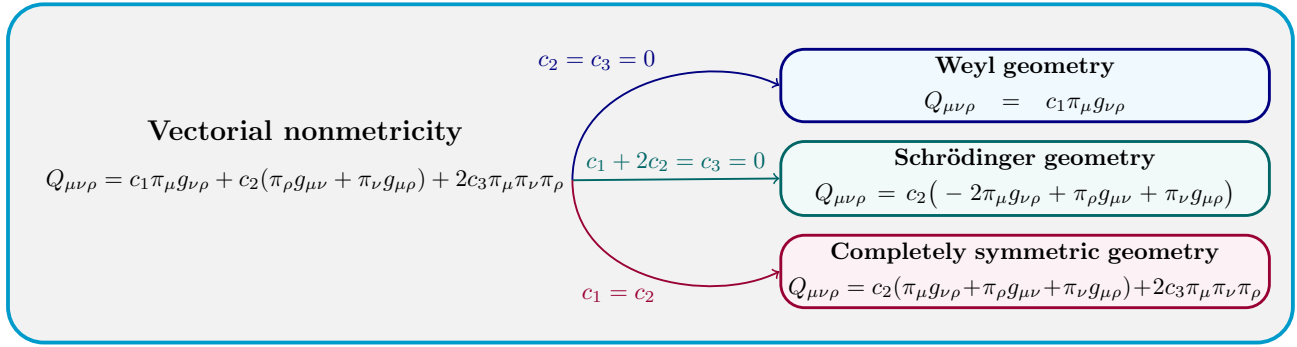
\begin{figure}[htbp]
\centering
\scalebox{0.7}{
\begin{tikzpicture}[
    node distance=0.1cm, 
    every node/.style={align=center}, 
    labelnode/.style={font=\normalsize, inner sep=2pt},
    mainbox/.style={
        draw=cyan!80!black,
        ultra thick,
        rounded corners=15pt,
        fill=gray!10,
        inner sep=2.5cm
    }
]
    \node[mainbox] at (0,0) {~~~~~~~~~~~~~~~~~~~~~~~~~~~~~~~~~~~~~~~~~~~~~~~~~~~~~~~~~~~~~~~~~~~~~~~~~~~~~~~~~~~~~~~~~~~~~~~~~~~~~~~~~~~~~~~~~~~~~~~~~~~~~~~~~}; 

    \node[font=\bfseries \large] (title) at (-5.1,0.6) {Vectorial nonmetricity};

    \node[font=\normalsize, below=of title, yshift=0cm] (mainEq) 
    {$Q_{\mu\nu\rho} = c_1 \pi_{\mu} g_{\nu\rho} + c_2 (\pi_{\rho} g_{\mu\nu} + \pi_{\nu} g_{\mu\rho}) + 2c_3 \pi_{\mu} \pi_{\nu} \pi_{\rho}$};
    
    \node[font=\itshape, right=of mainEq, yshift=0cm] (special) {~~~~~~~~~~~~~~};

    \node[
    draw=blue!50!black, 
    very thick, 
    rounded corners=10pt,
    fill=cyan!5, 
    right=of special, xshift=1cm,
    yshift=1.4cm, 
    text width=7cm
    ,font=\normalsize] (weyl) {
        \textbf{Weyl geometry}\\[0.3em]
        $Q_{\mu\nu\rho} = c_1 \pi_{\mu} g_{\nu\rho}$
    };
    
    \node[
    draw=teal!80!black, 
    very thick, 
    rounded corners=10pt,
    fill=teal!5, 
    below=of weyl, 
    yshift=-0.15cm, 
    text width=7cm, font=\normalsize
    ] (schrodinger) {
        \textbf{Schr\"{o}dinger geometry}\\[0.3em]
        $Q_{\mu\nu\rho} = c_2 \big( -2\pi_{\mu} g_{\nu\rho} + \pi_{\rho} g_{\mu\nu} + \pi_{\nu} g_{\mu\rho} \big)$
    };

    \node[
    draw=purple!80!black, 
    very thick, 
    rounded corners=10pt,
    fill=purple!5, 
    below=of schrodinger, 
    yshift=-0.15cm, 
    text width=7cm, font=\normalsize
    ] (symmetric) {
        \textbf{Completely symmetric geometry}\\[0.3em]
        $Q_{\mu\nu\rho} = c_2 (\pi_{\mu} g_{\nu\rho} + \pi_{\rho} g_{\mu\nu} + \pi_{\nu} g_{\mu\rho}) + 2c_3 \pi_{\mu} \pi_{\nu} \pi_{\rho}$
    };

   % Blue arrow - more upward curve
\draw[->, thick, blue!50!black] 
([xshift=-0.1cm]special.west)  
.. controls +(0,1.5) and +(-1,0.5) .. 
node[above left, pos=0.55, font=\normalsize] {$c_2 = c_3 = 0$} 
(weyl.west);

% Green arrow - straight
\draw[->, thick, teal!80!black] 
([xshift=-0.1cm]special.west)-- 
node[above, pos=0.5, font=\normalsize] {$c_1 + 2c_2 = c_3 = 0$}
(schrodinger.west);

% Purple arrow - more downward curve
\draw[->, thick, purple!80!black] 
([xshift=-0.1cm]special.west) 
.. controls +(0,-1.5) and +(-1,-0.5) .. 
node[below left, pos=0.55, font=\normalsize] {$c_1 = c_2$}
(symmetric.west);

\end{tikzpicture}

}

\caption{Illustration of vectorial nonmetricity and its special cases.}
\label{fig:vectorial_nonmetricity}
\end{figure}

The decomposition given in equation \eqref{Q_decomp} can also be applied to the vectorial nonmetricity defined in equation \eqref{vectorial_nonmetricity}. For such a connection, the trace (dilation) and traceless (shear) parts of the nonmetricity take the form
\begin{equation}
    Q_{\mu}=\left( c_1 + \frac{c_2}{2} \right) \pi_{\mu} + \frac{c_3}{2} \pi_{\mu} \pi_\lambda \pi^\lambda \, , \quad \overline{Q}_{\mu \nu \rho}=c_2 \pi_{\rho} g_{\mu \nu} + c_{2} \pi_{\nu} g_{\rho \mu}+ 2 c_3 \pi_{\mu} \pi_{\nu} \pi_{\rho} - \frac{c_2}{2} \pi_{\mu} g_{\nu \rho}-\frac{c_3}{2} g_{\nu \rho} \pi_{\mu} \pi_{\lambda} \pi^{\lambda} \ ,
\end{equation}
from which it can be easily seen that $c_1$ contributes exclusively to the dilation part, while $c_2$ and $c_3$ generate contributions to both dilation and shear. Table \ref{table:dilationshear} summarises the explicit forms of the dilation vector $Q_{\mu}$ and shear tensor $\overline{Q}_{\mu \nu \rho}$ for the three special cases illustrated in Fig.~\ref{fig:vectorial_nonmetricity}. Both the Schr\"{o}dinger and completely symmetric geometries have non-vanishing dilation and shear. When coupled to matter, both of these currents could have interesting phenomenological implications. For instance, while the dilation current relates to scale invariance, the shear current has been hypothesised to encode information about the internal structure of hadrons \cite{Hehl:1977gn,Hehl:1978cb,Iosifidis:2020upr}. While hypermomentum couplings are not investigated in this work, they represent an interesting direction for future study.

\begin{table}[ht]
\centering
\renewcommand{\arraystretch}{1.8}
\small
\resizebox{\textwidth}{!}{%
\begin{tabular}{|c|c|c|}
\hline
\textbf{Geometry} & \textbf{Dilation}  & \textbf{Shear}  \\
\hline
Weyl & \( c_1\pi_\mu \) 
& \( 0 \) \\
\hline
Schrödinger & \( -\frac{3}{2}c_2 \pi_\mu \) 
& \( c_2  \pi_{\rho} g_{\mu \nu} + c_2 \pi_{\nu} g_{\rho \mu}  - \frac{c_2}{2} \pi_{\mu} g_{\nu \rho} \) \\
\hline
Completely symmetric & 
\( \frac{3c_2}{2}  \pi_\mu + \frac{c_3}{2} \pi_\mu \pi_\lambda \pi^\lambda \) 
& \( c_2 \pi_{\rho} g_{\mu \nu} + c_{2} \pi_{\nu} g_{\rho \mu}+ 2 c_3 \pi_{\mu} \pi_{\nu} \pi_{\rho} - \frac{c_2}{2} \pi_{\mu} g_{\nu \rho}-\frac{c_3}{2} g_{\nu \rho} \pi_{\mu} \pi_{\lambda} \pi^{\lambda} \) \\
\hline
\end{tabular}
}
\caption{Explicit forms of the dilation and shear for the three special geometries.}
\label{table:dilationshear}
\end{table}

For the general vectorial nonmetricity given in
(\ref{vectorial_nonmetricity}), the distortion tensor (\ref{distortion}) takes the form
\begin{equation} \label{Ktor}
    N^{\lambda}{}_{\mu \nu} = \frac{2c_2-c_1}{2} g_{\mu \nu} \pi^{\lambda} + c_1 \delta_{(\mu}^{\lambda} \pi_{\nu)} + c_3 \pi_{\mu} \pi_{\nu} \pi^{\lambda} \, .
\end{equation}
A direct but lengthy calculation shows that the Riemann tensor is given by
\begin{equation}\label{curvaturegeneral}
    \begin{aligned}
        \tensor{R}{^\rho _\mu _\sigma _\nu}&= \overset{\circ}{R} \tensor{}{^\rho _\mu _\sigma _\nu} + \frac{c_1}{2} \delta^{\rho}_{\nu} \left( \overset{\circ}{\nabla}_{\sigma} \pi_\mu -\frac{c_1}{2} \pi_\sigma \pi_\mu -c_3 \pi^\beta \pi_\beta \pi_\sigma \pi_\mu \right) +\frac{c_1}{2} \delta^{\rho}_{\mu} \left( \overset{\circ}{\nabla}_{\sigma} \pi_\nu - \overset{\circ}{\nabla}_{\nu} \pi_{\sigma} \right)\\
        &+\frac{c_1}{2} \delta^{\rho}_{\sigma} \left( - \overset{\circ}{\nabla}_{\nu} \pi_{\mu} +\frac{c_1}{2} \pi_\nu \pi_\mu + c_3 \pi^\beta \pi_\beta \pi_\nu \pi_\mu \right) + c_3 \left(\overset{\circ}{\nabla}_{\sigma}\left( \pi^\rho \pi_\nu \pi_\mu \right) - \overset{\circ}{\nabla}_{\nu} \left( \pi^\rho \pi_\sigma \pi_\mu \right) \right)\\
        &+\frac{2c_2-c_1}{2} g_{\mu \nu} \left( \overset{\circ}{\nabla}_{\sigma} \pi^{\rho}+c_2 \pi^\rho \pi_\sigma + \frac{c_1}{2} \delta^{\rho}_{\sigma} \pi_\beta \pi^\beta + c_3 \pi^\beta \pi_\beta \pi^\rho \pi_\sigma - \frac{c_1}{2} \pi^\rho \pi_\sigma \right)\\
        &-\frac{2c_2-c_1}{2} g_{\sigma \mu} \left( \overset{\circ}{\nabla}_{\nu} \pi^\rho + \frac{2c_2-c_1}{2} \pi^\rho \pi_\nu + \pi_\beta \pi^\beta \left(\frac{c_1}{2} \delta^{\rho}_{\nu} +c_3 \pi^\rho \pi_\nu \right) \right)\, .
    \end{aligned}
\end{equation}
From this expression, the Ricci scalar is found to be

\begin{equation} \label{R_expand}
    R = \overset{\circ}{R} +  3\pi_{\mu} \pi^{\mu} \left(-\frac{c_1^2}{2}+c_1 c_2 + c_2^2+ c_2 c_3 \pi_\lambda \pi^\lambda \right) + 3\overset{\circ}{\nabla}_{\mu} \pi^{\mu} \left(-c_1+c_2 \right)  \, ,
\end{equation} 
with the final two terms being total derivatives. The properties of this expression will be discussed in the next section.

\section{Action principles for integrable nonmetricity}\label{sec:3}

The task of constructing actions in non-Riemannian geometries which lead to theories beyond standard GR is challenging for a number of reasons. 
For instance,  simply taking the affine Ricci scalar and solving the connection field equations will lead back to the Levi-Civita geometry. The long-standing problem of finding an action which leads to the metric-affine Einstein tensor in non-Riemannian geometries is therefore unresolved, see for instance \cite{Klemm:2020mfp}.
Consequently, the vast majority of non-Riemannian modifications take one of two approaches:
\begin{enumerate}
    \item Construct higher-order actions from the geometric quantities $R^{\mu}{}_{\nu \rho \sigma}$, $T^{\mu}{}_{\nu \rho}$, $Q_{\mu \nu \rho}$ and their derivatives \cite{Iosifidis:2021bad}. 
    \item Assume matter-geometry couplings, giving rise to non-vanishing hypermomentum \cite{Hehl:1994ue}.
\end{enumerate}
In the first case, the connection field equations become genuinely dynamical (i.e., differential), whereas in the second, they are algebraic but include source terms. Both scenarios can lead to deviations from GR with
non-vanishing torsion and nonmetricity. A concrete example is provided by the linear vectorial distortion models studied in \cite{BeltranJimenez:2015pnp,BeltranJimenez:2016wxw}, which arise from quadratic gravitational actions.

However, there are challenges associated with both of these routes. The first, involving higher-order geometric actions, has been extensively studied in the case of $f(R)$ gravity and its extensions \cite{Sotiriou:2006hs,Sotiriou:2006qn,DeFelice:2010aj,Wu:2018idg}. When confronting these theories with data, it is not surprising that some are in tension \cite{Amarzguioui:2005zq}, while others show better agreement \cite{Gomes:2023xzk}. Nevertheless, a general comparison with observational data is challenging given the infinite freedom in the choice of model. This makes the theory space as a whole very difficult to constrain, regardless of the geometric setting. Moreover, higher-order actions are well-known to be prone to instabilities. Even in the quadratic case, the conditions for a healthy spectrum on FLRW backgrounds are highly non-trivial \cite{Aoki:2023sum}, and studying stability conditions on general backgrounds is a complicated task (e.g., see \cite{Bahamonde:2024efl} and \cite{Bahamonde:2024zkb} for studies in cubic metric-affine and quadratic teleparallel theories respectively).

The second approach with matter-geometry couplings is potentially problematic due to the non-conservation of the energy-momentum tensor and the non-geodesic motion of particles \cite{Velten:2021xxw,Bahamonde:2021akc}. The choice of matter action must also be well-motivated and give rise to non-trivial hypermomentum \cite{Hehl:1976kv,Hehl:1994ue}. For instance, standard scalar and vector fields generically have vanishing hypermomentum, while fermions (Dirac spinors) couple only to torsion. It follows that for Einstein-Weyl geometries, at least, one must either add additional nonminimal couplings or use non-standard matter sources to generate non-trivial hypermomentum. Although further work is needed to assess the viability and observational implications of these models, progress has been made in this direction \cite{Iosifidis:2020gth,Iosifidis:2021nra,Iosifidis:2023kyf,Iosifidis:2024ksa}.

In this work, we take a third approach which does not rely on higher-order actions nor matter-geometry couplings. Instead, we assume the nonmetricity vector is integrable, as described in the previous section. This allows us to construct the gravitational action using only the Ricci scalar, which now includes contributions from both curvature and integrable vectorial nonmetricity. While this approach is far less studied, it has been utilised effectively in the Weyl geometric framework \cite{Miritzis:2013ai,Paliathanasis:2020plf,Paliathanasis:2021qns,Salim:1996ei}. To our knowledge, it has not been considered in the broader case of general vectorial nonmetricity such as \cite{aringazin1991matter}. We show that the combination of integrability with the extended form of vectorial nonmetricity (\ref{vectorial_nonmetricity}) is crucial for cosmology. These new terms are necessary for late-time accelerated expansion in theories constructed linearly from the Ricci scalar. To implement the integrability condition we introduce Lagrange multipliers, verifying the consistency of the model. From this point, a natural equivalence with scalar-tensor theories emerges, which we discuss in detail. This language is then utilised to determine stability constraints for our model, to be used in the following cosmological analysis.

\subsection{Field equations and variations}\label{sec:3A}

Our gravitational action will be constructed from the Ricci scalar including the effects of both curvature and vectorial nonmetricity (\ref{R_expand}). The quartic nonmetricity vector terms have not appeared in the literature, being absent for the well-studied case of Weyl spacetimes with $c_2 = c_3 = 0$ \cite{Salim:1996ei,Quiros:2022uns,Miritzis:2013ai} and in the more general cases with linear vectorial nonmetricity with $c_3=0$ \cite{aringazin1991matter,BeltranJimenez:2015pnp}. From Fig.~\ref{fig:vectorial_nonmetricity}, we see that they are only present for the completely symmetric geometry.
These new terms can be thought of as the next leading-order contributions in powers of $\pi$. Higher order terms such as $\mathcal{O}(\pi^6)$ would require $Q_{\mu \nu \lambda}$ itself to include contractions $\pi^{\mu}\pi_{\mu}$, which seems considerably less natural and is therefore not considered in this work. Alternatively, one could construct actions from higher powers of the Riemann or nonmetricity tensors, but for simplicity we restrict our attention to actions linear in the Ricci scalar.

We now postulate that the vectorial component is integrable $\pi_{\mu} = \partial_{\mu} \phi$, where $\phi$ is an arbitrary scalar field. This condition can be implemented by supplementing the gravitational action with the Lagrange multiplier term
\begin{equation} \label{lagmult}
    S_{\lambda} = \frac{1}{2\kappa} \int \sqrt{-g} \lambda^{\mu}\big( \pi_{\mu} - \overset{\circ}{\nabla}_{\mu} \phi \big)  d^4x  \, ,
\end{equation}
which, we note, has not appeared in the literature before.
Lastly, to establish a direct correspondence with integrable Weyl-type theories \cite{Salim:1996ei}, we include the non-trivial vectorial divergence term 
\begin{equation}
     \nabla_{\mu} \pi^{\mu} = \big(2 c_1 + c_2 \big) \pi_{\mu} \pi^{\mu} + c_3\pi_{\mu} \pi^{\mu}\pi_{\nu} \pi^{\nu} + \overset{\circ}{\nabla}_{\mu} \pi^{\mu} \, ,
\end{equation}
multiplied by an arbitrary constant $\xi$, see also \cite{Novello:2008ra,Miritzis:2013ai}.

The total action for the integrable vectorial nonmetricity theory is given by 
\begin{align} 
    S[g,\pi,\phi,\lambda,\psi] = \frac{1}{2\kappa} \int \sqrt{-g} \big( R + \xi \nabla_{\mu} \pi^{\mu}  \big) d^4x + S_{\lambda}[g,\pi,\phi,\lambda] + S_{\textrm{m}}[g,\psi]   \, , \label{S_main}
\end{align}
where $S_{\textrm{m}}$ is the minimally coupled matter action with matter fields $\psi$. No further geometry-matter couplings are assumed beyond the standard minimal metric couplings, for the reasons discussed above.
The variations with respect to the Lagrange multiplier lead to the integrable condition $\pi_{\mu} = \partial_{\mu} \phi$, while the $\phi$ variations enforce the covariant conservation of the Lagrange multiplier $\overset{\circ}{\nabla}_{\mu} \lambda^{\mu} =0$. The latter is used to eliminate $\lambda^{\mu}$ from the vectorial field equation after taking the divergence, see Appendix \ref{append_var} for  full details.
The final metric and scalar field equations can be written as
\begin{align}  \label{EoM_met}
    \overset{\circ}{G}_{\mu \nu} + \overset{\circ}{\nabla}_{\mu} \phi \overset{\circ}{\nabla}_{\nu} \phi \Big(b_1 + b_2 \overset{\circ}{\nabla}_{\lambda} \phi \overset{\circ}{\nabla} \tensor{}{^{\lambda}} \phi \Big) - \frac{1}{2} g_{\mu \nu}\overset{\circ}{\nabla}_{\lambda} \phi \overset{\circ}{\nabla} \tensor{}{^{\lambda}} \phi \Big(b_1 + \frac{b_2}{2} \overset{\circ}{\nabla}_{\rho} \phi \overset{\circ}{\nabla}\tensor{}{^{\rho}} \phi \Big) &= \kappa T_{\mu \nu} \, , \\
    b_1 \overset{\circ}{\nabla}_{\mu} \overset{\circ}{\nabla} \tensor{}{^{\mu}} \phi + b_2 \Big( \overset{\circ}{\nabla}_{\mu} \phi \overset{\circ}{\nabla}\tensor{}{^{\mu}}\phi \overset{\circ}{\nabla}_{\nu} \overset{\circ}{\nabla}\tensor{}{^{\nu}} \phi + 2  \overset{\circ}{\nabla}_{\mu} \phi\overset{\circ}{\nabla}\tensor{}{^{\nu}} \phi \overset{\circ}{\nabla}_{\nu} \overset{\circ}{\nabla}\tensor{}{^{\mu}} \phi \Big) &= 0 \, , \label{EoM_phi}
\end{align}
where the following two constants parameterise the whole theory
\begin{equation} \label{bs}
    b_1 = - \frac{3 c_1^2}{2} + c_2 (3c_2 + \xi) + c_1 (3c_2 + 2 \xi) \, , \qquad b_2 = 2 c_3 (3c_2 + \xi) \, .
\end{equation}
Here, the constants $c_1,c_2$ and $\xi$ are dimensionless\footnote{This follows from dimensional analysis: nonmetricity $Q_{\mu \nu \rho}$ has dimensions of length, the vector $\pi_{\mu}=\partial_{\mu} \phi$ has dimension $L^{-1}$, and the metric $g_{\mu \nu}$ has dimension $L^2$. Hence, $c_1$ and $c_2$ must be dimensionless. Since the cubic term $c_3 \pi_\mu \pi_\nu \pi^\lambda$ must match the dimension of $\tensor{Q}{^\mu _\nu _\lambda} \sim L$, it follows that $[c_3]=L^2$.}, while $c_3$ has dimension  $L^{2}$. This will be important when constructing dynamical variables in the following section.
Since matter is minimally coupled, the energy-momentum tensor is covariantly conserved \cite{Koivisto:2005yk}
\begin{equation}
    \overset{\circ}{\nabla}\tensor{}{_{\mu}}T^{\mu \nu} = 0 \, .
\end{equation}
This conservation law can be explicitly verified by taking the covariant divergence of the metric field equation (\ref{EoM_met}) and invoking the scalar field equation (\ref{EoM_phi}).

For the three geometric settings described in the previous section, the corresponding field equations can be obtained by evaluating the constants $b_1$ and $b_2$. In the Weyl case, we leave $\xi$ arbitrary in order to reproduce past results such as \cite{Novello:2008ra,Miritzis:2013ai}, but for the other geometric cases we set it to zero. 
In the following section, we give more precise details relating to the stability of these geometries.
\begin{itemize}
    \item[$\triangleright$] \textit{Weyl geometry} with $c_2=c_3=0$ and arbitrary $\xi$ gives 
    \begin{equation} \label{Weyl}
        b_1 = -\frac{3c_1^2}{2} + 2 c_1 \xi \, , \quad b_2 = 0 \, .
    \end{equation}
    The kinetic term in the action has the correct (negative) sign if either $c_1 >0 $ and $ 4\xi < 3 c_1$, or if $c_1 <0 $ and $ 4\xi > 3 c_1$. The theory is mapped to standard gravity minimally coupled to a canonical scalar field ($b_1 = -1/2)$ with vanishing potential when $\xi =(3c_1^2-1)/4c_1$. More details on the scalar-tensor representation of the Weyl geometries and further possible extensions can be found in \cite{Novello:2008ra,Pucheu:2016act,Chatzidakis:2022mpf}.
    \item[$\triangleright$] \textit{Schr\"{o}dinger geometry} with $c_1 = -2 c_2$, $c_3=0$ and $\xi=0$ gives
    \begin{equation}
        b_1 = -9 c_2^2 \, , \quad b_2 = 0 \,  .
    \end{equation}
    The parameter $b_1$ is always negative and the kinetic term therefore always has the correct sign. 
    \item[$\triangleright$] \textit{Completely symmetric geometry} with $c_1=c_2$ and $\xi=0$ gives
    \begin{equation} \label{symm}
        b_1 = \frac{9 c_2^2}{2} \, , \quad b_2 = 6 c_2 c_3  \, .
    \end{equation}
    In this case, both $b_1$ and $b_2$ are generally non-zero. The coefficient $b_1$ is always positive, whereas $b_2$ can take either sign depending on the parameters. In the next section, we study the stability conditions associated with having a kinetic term $b_1$ with the `incorrect' (positive) sign.
\end{itemize}

Before proceeding to the cosmological analysis, it is illustrative to examine the scalar-tensor equivalent representation of the theory. This reveals a fundamental duality between the integrable vectorial nonmetricity geometries and a popular class of non-canonical scalar-tensor models.
It has been argued that the constraints one imposes from the scalar-tensor representation may not necessarily hold the same weight in the geometric setting \cite{Miritzis:2013ai,Pucheu:2016act}. Nonetheless, we will find these constraints to be extremely useful when finding viable cosmological solutions.

\subsection{Scalar-tensor representation}
Rewriting our initial Lagrangian in a scalar-tensor form allows us to compare with previously studied models and make some general claims about perturbative stability. Using the integrable vectorial nonmetricity $\pi_{\mu} = \overset{\circ}{\nabla}_{\mu} \phi$ and the parametrisation introduced in (\ref{bs}), the Lagrangian (\ref{S_main}) can be expressed as 
\begin{equation} \label{Ricciscalarplusxi}
    R + \xi \nabla_{\mu} \pi^{\mu} = \overset{\circ}{R} + b_1 \overset{\circ}{\nabla}\tensor{}{^{\mu}}\phi \overset{\circ}{\nabla}_{\mu} \phi + \frac{1}{2} b_2\overset{\circ}{\nabla}\tensor{}{^{\mu}} \phi \overset{\circ}{\nabla}_{\mu}\phi \overset{\circ}{\nabla}\tensor{}{^{\nu}}\phi \overset{\circ}{\nabla}_{\nu}\phi + b_3 \overset{\circ}{\nabla}_{\mu} \overset{\circ}{\nabla}\tensor{}{^{\mu}}\phi \, ,
\end{equation}
where the final boundary term has an irrelevant prefactor $b_3 = -3c_1 + 3 c_2 + \xi$. Hence, the part of the Lagrangian which contributes to the dynamical equations depends only on the two parameters $b_1$ and $b_2$, as is evident from the field equations (\ref{EoM_met})-(\ref{EoM_phi}).
In the case that $b_2=0$, the theory resembles Einstein gravity minimally coupled to a scalar field with a kinetic term $\propto b_1$. For a healthy kinetic term, one may naively expect $b_1<0$. However, it has been argued that because $\phi$ is geometric in nature, there are no restrictions on the sign of this term \cite{Miritzis:2013ai}. On the other hand, models with $b_1 >0$ correspond to well-known ghost condensate models \cite{Arkani-Hamed:2003pdi}. Moreover, such models can be consistently viewed as low-energy effective theories \cite{Babichev:2018twg}, giving further reason to consider these parameter choices as potentially physical.
Without a priori assuming the sign of either $b_1$ or $b_2$, we will go on to investigate the dynamical stability of cosmological solutions for both choices.
Lastly, we point out that extensions of (\ref{Ricciscalarplusxi}) within the integrable Weyl geometry ($b_2=0$) with an additional potential term $V(\phi)$ have also been studied, see for instance \cite{Pucheu:2016act}. However, from the geometric point of view, $\phi$ is related to the integral of nonmetricity and therefore including these terms does not appear natural in our setting. For this reason, we neglect any potential $V(\phi)$.

Inspection of our Lagrangian (\ref{Ricciscalarplusxi}) reveals a direct equivalence with a class of `purely kinetic' $k$-essence models \cite{Armendariz-Picon:1999hyi,Armendariz-Picon:2000ulo}. As this structure emerges directly from our underlying spacetime geometry, we will also refer to the model as \textit{geometric $k$-essence}. 
Characterising the scalar field part of the action in the function $P(\phi,X)$ where $X = -\overset{\circ}{\nabla}\tensor{}{^{\mu}}\phi \overset{\circ}{\nabla}_{\mu} \phi/2 = -\pi_{\mu} \pi^{\mu}/2$, the equivalent geometric $k$-essence Lagrangian is simply
\begin{equation} \label{Pmodel}
    \mathcal{L}_{k} = P(\phi,X) = -b_1 X + b_2 X^2 \, ,
\end{equation}
where the factor of $1/2$ in (\ref{S_main}) has been taken into account. This specific quadratic form of kinetic $k$-essence action has been well-studied, appearing in the first proposed kinetically driven inflationary paradigm \cite{Armendariz-Picon:1999hyi} (see also \cite{Chiba:1999ka,Scherrer:2004au,Armendariz-Picon:2000ulo,dePutter:2007ny} for further historic studies). The parameters $b_1$ and $b_2$ are often assumed to be non-negative, but we will consider the whole parameter space for our dynamical analysis\footnote{Usually $\phi$ is rescaled to absorb one of the parameters $b_1$ or $b_2$, but this is only possible when $b_1$ and $b_2$ have a fixed sign. After the dynamical systems analysis, once the signs of $b_1$ and $b_2$ are determined via stability arguments, we will perform this rescaling.}.
The signs of these two parameters also turn out to be relevant in astrophysical contexts, such as in relation to screening mechanisms \cite{Bezares:2021dma,Bezares:2021yek}, and a comparison of the theoretical constraints on $b_1$ and $b_2$ in different regimes would be interesting to investigate in the future.

The quadratic power-law model (\ref{Pmodel}) can be obtained by simply considering the first two (non-constant) terms of the Taylor expansion of a generic function $P(X)$ with $P \rightarrow 0$ as $X \rightarrow 0$. It then becomes apparent that had we included higher-order contractions of the vector $\pi^{\mu}\pi_{\mu}$ in our nonmetricity ansatz (\ref{vectorial_nonmetricity}), the next leading-order terms in the Taylor expansion of $P(X)$ would be obtained. This would allow one to construct any purely kinetic, power-law $k$-essence model from an action constructed from the Ricci scalar with integrable vectorial nonmetricity. In this work, we focus only on the quadratic models.

%%%% 

The equation of state for the $k$-essence fluid is given by
\cite{Armendariz-Picon:1999hyi}
\begin{equation} \label{w_phi}
    w_{\phi} = \frac{P}{\rho_{\phi}} = \frac{P}{2X P_{,X}-P } = \frac{b_1 - b_2 X}{ b_1 - 3 b_2 X} \, ,
\end{equation}
where we have identified $P$ as the pressure and the energy density as $\rho_{\phi} = 2X P_{,X}-P$. For $b_2 =0$ we have the kinetic-dominated canonical scalar field equation of state $w_{\phi} = 1$, while for $b_1=0$ we instead have $w_{\phi} = 1/3$. The latter is the radiation tracking solution \cite{Armendariz-Picon:2000nqq}. 
The adiabatic sound speed of the scalar field  perturbations is \cite{Armendariz-Picon:1999hyi}
\begin{equation} \label{cs2}
    c_s^2 = \frac{P_{,X}}{\rho_{\phi,X}}  = 
    \frac{b_1 - 2b_2 X}{b_1 - 6 b_2 X} = \frac{1+w_{\phi}}{5-3w_{\phi}} \, .
\end{equation}
For the two cases $b_1=0$ or $b_2=0$, the sound speed is equal to the background equation of state $c_s^2 =w_{\phi} =1/3$ and $c_s^2 =w_{\phi} =1$  respectively. In more general cases, the sound speed varies and is unbounded from below and above.

It is well known that non-canonical scalar field Lagrangians can introduce pathologies at both the classical and quantum level. At the classical level, stability under small perturbations and the absence of ghosts imposes consistency conditions on the model. These can be implemented by requiring a non-negative sound speed squared\footnote{Note that imaginary sound speeds have also been studied in $k$-essence scenarios \cite{Bouhmadi-Lopez:2016cja}. In an EFT approach, the usual gradient instabilities with $c_s^2<0$ are regulated below the cutoff by higher-derivative operators that modify the dispersion relation; see the original ghost condensate construction \cite{Arkani-Hamed:2003pdi} and related EFT discussion in \cite{deRham:2017aoj}. A concrete UV completion in the $P(X)$ models is given in \cite{Babichev:2018twg}.} and a positive derivative of the energy density for the $k$-essence fluid \cite{Garriga:1999vw,Armendariz-Picon:2000ulo},
\begin{equation} \label{constr}
    c_s^2 \geq 0   \quad \textrm{and} \quad  \rho_{\phi,X} >  0  \, .
\end{equation}
For the purely kinetic model with $P(X)=-b_1X+b_2X^2$, these conditions give
\begin{equation}
2b_2 X \geq b_1, \; \; \text{and} \; \; 6b_2 X >b_1.
\end{equation}
Given that on cosmological backgrounds the kinetic term $X$ is non-negative, there are four distinct cases required for stability:
\begin{equation}
\label{constr2}
{
\left\{
\begin{aligned}
&\text{$b_2>0$ and $b_1 > 0$ with $X \geq \frac{b_1}{2b_2}$} \, , \\
&\text{$b_2 > 0$ and $b_1 \leq 0$} \, , \\
&\text{$b_2<0$ and $b_1 < 0$ with $X < \frac{b_1}{6b_2}$} \, , \\
&\text{$b_2=0$ and $b_1<0$}\,  .
\end{aligned}
\right.
}
\end{equation}
The two branches with either $b_1 >0$ or $b_1 \leq 0$ display uniquely different physical properties, which we discuss below. Note that the case with $b_2<0$ and $b_1>0$ is always unstable.

A further condition sometimes imposed for viable background cosmological dynamics is a positive energy density $\rho_{\phi} \geq 0$ \cite{Bahamonde:2017ize}. Such a condition is usually independent from the stability criteria (\ref{constr}); however, for our particular quadratic model, if both $\rho_{\phi} \geq 0$ and $c_s^2 \geq 0$ are upheld,
this gives identical constraints to (\ref{constr2}). We therefore use $\rho_{\phi} \geq 0$ interchangeably with $\rho_{\phi,X} \geq 0$ in this work. Note that $b_1 < 0$ also arises from requiring the quadratic kinetic term appearing in the action $b_1 \pi^{\mu} \pi_{\mu} = - 2b_1 X$ to have the correct sign, see Sec. \ref{sec:3A}. Hence, both the Weyl and Schr\"{o}dinger geometries (with $\xi=0$ and $b_2=0$) satisfy the first set of stability constraints. The result is less trivial in the case where the quadratic kinetic terms are present in the action, $b_2 \neq 0$, such as the completely symmetric geometries.

In the branch\footnote{Note that in the case with $b_1 < 0$ and $b_2<0$, we retrieve the well-known result that all stable configurations are superluminal \cite{Babichev:2007dw}. However, in the next section we show that these regions evolve to past instabilities.} where $b_1  \leq 0$ the equation of state is bound between $1/3 \leq w_{\phi} < 5/3$, while in the other branch with $b_1>0$ we instead have $-1 \leq w_{\phi} \leq 1/3$. The former $b_1 \leq 0$ branch therefore cannot give rise to stable accelerating solutions. Consequently, neither branch can allow for stable phantom crossing, see also \cite{Vikman:2004dc}. More importantly, in the physically interesting $b_1>0$ case, 
the sound speed is free from superluminal  propagation at all times and bound between $0 \leq c_s^2 < 1/3$. However, a more detailed analysis of the full phase space is needed to ensure that these unstable regions with $c_s^2 <0$ or $\rho_{\phi} <0 $ are dynamically avoided at all times (i.e., along trajectories within the physical phase space). We will investigate exactly this in the dynamical systems analysis of the following section. We will also explain why these additional stability arguments are essential for constraining the model against observations.

\section{Cosmological dynamics of geometric k-essence}\label{sec:4}

Having introduced the key theoretical properties and geometric origin of the purely kinetic $k$-essence model, we now turn to cosmological applications. Recent cosmological studies of various $k$-essence models have shown a number of interesting features, such as early dark energy behaviour \cite{Tian:2021omz} and the ability to regularise cosmological singularities \cite{FerreiraJunior:2023qxi,Huang:2022hye,Das:2023umo}. They have also shown promise in the possibility of alleviating cosmic tensions \cite{Huang:2021urp,Copeland:2023zqz,HosseiniMansoori:2024pdq}. However, these works have also included potential terms, which are not feasible from our geometric standpoint. The purely kinetic models, and in particular the quadratic models (\ref{Pmodel}), have received far less attention. Recent works on these models \cite{Hussain:2024qrd,Quiros:2025fnn} tend to fix the signs of the free parameters $b_1$ and $b_2$ in their analysis for phenomenological reasons. However, as is clear from equations (\ref{Weyl})-(\ref{symm}), this would rule out some of the geometric frameworks introduced in Section \ref{sec:2}. We therefore leave $b_1$ and $b_2$ arbitrary for now, working in as much generality as possible.

The spatially flat Friedmann–Lema\^{i}tre–Robertson–Walker (FLRW) metric of a homogeneous and isotropic background is given by 
\begin{equation}
    ds^2 = -dt^2 + a(t)^2 \delta_{ij} dx^{i} dx^j \, .
\end{equation}
On these backgrounds, nonmetricity is generically composed of three independent functions of time \cite{Minkevich:1998cv}. One arrives at our specific vectorial nonmetricity (\ref{vectorial_nonmetricity}) by fixing two of these functions and making straightforward redefinitions\footnote{It is interesting to note that the completely general nonmetricity tensor satisfying the FLRW symmetries (see Eq.~(70) of \cite{Iosifidis:2020gth}) takes the same form as our specific vectorial nonmetricity in (\ref{vectorial_nonmetricity}). In fact, promoting the two constant parameters $c_1$ and $c_2$ to arbitrary functions of time and rescaling $\pi(t)$ directly reproduces the most general cosmological nonmetricity \cite{Minkevich:1998cv,Iosifidis:2020gth}. For further study of this type of nonmetricity on FLRW backgrounds we refer to \cite{Iosifidis:2020upr,Iosifidis:2021bad,Iosifidis:2022evi}.}. Furthermore, with the integrable condition, the nonmetricity vector is simply the time derivative of the scalar field $\pi_{\mu} = \dot{\phi}$. Given that our action is algebraic in nonmetricity, this leaves only a single degree of freedom $\phi$.
The dynamics for $\phi$ are determined from the scalar field equations (\ref{EoM_phi}), which are fundamentally derived from action principles. This is in contrast to many past approaches studying cosmological nonmetricity \cite{Iosifidis:2020upr,Iosifidis:2021bad,Iosifidis:2022evi,Csillag:2024eor}, where the dynamics are prescribed by hand at the level of equations of motion (equivalent to choosing a specific equation of state).
Similarly, in our approach the continuity equation follows from our action principles, without needing to be assumed or added by hand.

For the matter sector, we introduce a perfect fluid energy-momentum tensor
\begin{equation}
    T_{\mu \nu} = (\rho + p) u_{\mu} u_{\nu} + g_{\mu \nu} p \, ,
\end{equation}
with energy density $\rho$, pressure $p$ and four velocity $u^{\mu}$.
The cosmological equations then read
\begin{align}
    3H^2 &= \kappa \rho  - \frac{1}{2} b_1 \dot{\phi}^2 + \frac{3}{4} b_2 \dot{\phi}^4 \, , \\
        3 H^2 + 2\dot{H} &= - \kappa p
    + \frac{1}{2} b_1 \dot{\phi}^2 - \frac{1}{4} b_2 \dot{\phi}^4 \, ,
\end{align}
and the scalar field equation is
\begin{equation}
    b_1 \overset{\circ}{\Box} \phi - 3 b_2 \dot{\phi}^2\left(  \overset{\circ}{\Box}\phi  + 2 H \dot{\phi} \right) =0 \, ,
\end{equation}
where $\overset{\circ}{\Box} \phi = -3H \dot{\phi} - \ddot{\phi}$ is the Levi-Civita d'Alembert operator. 
It is then straightforward to confirm that the covariant conservation of matter follows from the cosmological and scalar field equations
\begin{equation} \label{cont_0}
    \dot{\rho} + 3 H(\rho + p ) = 0\, .
\end{equation} 
It is convenient to make the redefinition $\varphi(t) := \dot{\phi}$, after which the field equations can be stated as
\begin{align} \label{varphi1}
    3H^2 &= \kappa \rho - \frac{1}{2} b_1 \varphi^2 + \frac{3}{4} b_2 \varphi^4 \, , \\  \label{varphi2}
        3 H^2 + 2\dot{H} &= - \kappa p + \frac{1}{2} b_1 \varphi^2 - \frac{1}{4} b_2 \varphi^4 \, , \\ 
    b_1(3H\varphi+\dot{\varphi}) &= b_2 ( 3H \varphi^3 + 3 \varphi^2 \dot{\varphi} ) \, .  \label{varphi3}
\end{align}
Lastly, we give the scalar field pressure, energy density, equation of state and adiabatic sound speed squared in terms of the new field $\varphi^2 = \dot{\phi}^2 = 2X$, 
\begin{equation}
\begin{aligned} \label{eos_rel}
    p_{\varphi} &= -\frac{1}{2} b_1 \varphi^2 + \frac{1}{4} b_2 \varphi^4 \,, \quad & \rho_{\varphi} &= -\frac{1}{2} b_1 \varphi^2 + \frac{3}{4} b_2 \varphi^4 \, , \\ 
    w_{\varphi} &= \frac{2b_1 -  b_2 \varphi^2}{2b_1-3b_2 \varphi^2} \, , \quad &  c_s^2 &=  \frac{b_1 -  b_2 \varphi^2}{b_1-3b_2 \varphi^2}  \, .
\end{aligned}  
\end{equation}
This fully characterises the effective fluid resulting from linear actions with integrable vectorial nonmetricity in cosmological spacetimes. 

\subsection{Dynamical systems formulation}

Dynamical systems theory techniques have been successfully applied to many cosmological models, providing useful qualitative information about the background dynamics and evolution of physical parameters \cite{Bahamonde:2017ize,Coley:2003mj,Boehmer:2022wln}. In the case of the kinetic $k$-essence models, recent studies can be found in \cite{Fang:2014qga,Chakraborty:2019swx,Quiros:2025fnn,Hussain:2024qrd}. To date, however, there has been no fully satisfactory dynamical systems analysis of the purely kinetic quadratic models. Regardless of the choice of variables used, the resulting evolution equations and phase space contain problematic divergences. The points where the equations diverge can therefore not be trusted, unless further techniques or numerical methods are used. 
For instance, in the case at hand, applying linear stability fails at these points, with different eigenvalues being obtained depending on which trajectory is taken---as is well-known, a limit is not well-defined if different values can be obtained when approaching the point from different directions. Such limits appear to have been taken in past works, explaining some of the contradictory results found in the literature. 
Nonetheless, the numerical solutions around these points appear to be consistent, indicating that this is a problem with the formulation and not the underlying physical model. 
Similarly, the problems will not affect our physical results, provided these divergent points are avoided. A short discussion of past works and the disagreements in the literature is given in Appendix \ref{sec:comment}.

The dynamical systems analysis will be utilised for two main purposes: the first is to study the stability of the fixed point solutions; the second is to find accurate priors for the late-time observational analysis. These dynamical systems methods are less-often used to directly inform observational studies, but informative examples can be found in \cite{Amendola:2020ldb,Fritis:2023lhd}. Here we show this to be an essential step for the $k$-essence models, due to the existence of regions which violate the stability conditions  \eqref{constr}. The information gained from the stability analysis allows us to avoid these pathological regions and obtain physically realistic cosmological solutions. 

Dividing the Friedmann equation by $3H^2$ gives us the dimensionless Hubble constraint
\begin{equation}
    1 = \frac{\kappa \rho}{3H^2}  - \frac{b_1 \varphi^2}{6 H^2} + \frac{b_2\varphi^4}{4H^2} \, .
\end{equation} 
Note that $b_1$ is dimensionless while $b_2$ has dimension $H^{-2}$, which follows from (\ref{bs}).
Next, we introduce dimensionless variables
\cite{Chakraborty:2019swx}
\begin{align} \label{dyn_vars}
     \Omega_m = \frac{\kappa \rho}{3H^2} \, , \quad
     x = -\frac{b_1 \varphi^2}{6H^2}  \, ,  \quad
     y = \frac{b_2 \varphi^4}{4H^2}  \, .
\end{align}
The variables $x$ and $y$ are clearly related to one another, so one may wonder whether these are actually independent. However, if we define a  new variable  $\widetilde{y} = y/x^2$, we see that $\widetilde{y}=9 b_2 H^2/b_1^2$. This allows us to determine $H$ given solutions in $x$ and $y$. Alternatively, it shows that $x$ does not uniquely determine $y$ and vice-versa due to the presence of the unknown function $H$. In this sense, one variable can be related to the scalar field $\varphi$ and the other to the Hubble parameter $H$. In fact, a similar choice of variables was used in the recent study\footnote{In terms of the kinetic term $X$, the variables used in \cite{Quiros:2025fnn} are $x=b_1 X/(b_1 X+H^2)$ and $y=b_2 H^2 /(b_1^2 + b_2 H^2)$. This can be compared with our choice (\ref{dyn_X}) to give a mapping between both sets, provided $b_1$ and $b_2$ are both non-negative (as assumed in their work).} \cite{Quiros:2025fnn}, with one variable for $\varphi$ and the other for $H$. One difference in our dynamical systems analysis is that we allow $b_1$ and $b_2$ to be both positive and negative in order to construct a more general phase space and relate to our underlying geometry (\ref{Ktor}). We will then find constraints on the values of $\varphi_0$, $b_1$ and $b_2$ to use as priors in the late-time observational analysis of Section \ref{sec:5}. As a last note,  we can write the scalar field variables in terms of the kinetic term as
\begin{equation} \label{dyn_X}
    x = -\frac{b_1 X}{3H^2} \, , \qquad y = \frac{b_2 X^2}{H^2} \, .
\end{equation}

With our variables (\ref{dyn_vars}), the Hubble constraint is given by
\begin{equation} \label{hub_const}
    1 = \Omega_m + x + y  \, .
\end{equation}
Assuming a linear matter equation of state $p = w \rho$, we can use the acceleration equation to write the derivative of the Hubble factor as
\begin{equation} \label{Hdot}
    \frac{\dot{H}}{H^2} = -\frac{1}{2}\big(3+3x + y  + 3w \Omega_m \big) \, .
\end{equation}
We then find the effective equation of state $w_{\textrm{eff}}$ and deceleration parameter $q$ in terms of our dynamical variables,
\begin{align}
    w_{\textrm{{eff}}} &= x + \frac{y}{3} + w \Omega_m \, , \\
    q &:= -\frac{\dot{H}}{H^2}-1 = \frac{1}{2}\big(1+3x+y +3w\Omega_m \big) \, . \label{decel}
\end{align}
When matter dominates over the scalar field, $\Omega_m=1$ and $x= y =0$, and we obtain $w_{\textrm{{eff}}}=w$ and $q=(1+3w)/2$ as expected. 
It will also be important to state the scalar field equation of state (\ref{w_phi}) and its sound speed squared (\ref{cs2}) in terms of our new variables
\begin{equation} \label{cs2b}
    w_{\phi} = \frac{3 x +y}{3(x+y)} \, , \qquad c_s^2 = \frac{1}{3} + \frac{2 x}{3(x+2y)}\, .
\end{equation}
These will be used as additional stability constraints for the physical phase space and are essential for the viability of these models, see also \cite{Quiros:2025fnn}. Introducing the total scalar field density $\Omega_{\phi} = x+y$, we can then simply write 
\begin{equation}
    w_{\textrm{eff}} = w_{\phi} \Omega_{\phi} + w \Omega_m  \, .
\end{equation}
This confirms that during matter domination $\Omega_m=1$ with $x \rightarrow 0$ and $y \rightarrow 0$ the apparent divergence of the scalar field equation of state is not physical, with the correct physical quantity being $w_{\phi} \Omega_{\phi}$ which simply vanishes.
The origin $(0,0)$ must therefore be taken carefully to obtain the correct physical interpretation, which will show up in the fixed point analysis.

To derive the evolution equations we take the derivative with respect to $N = \log a$ and use the Friedmann constraint to eliminate $\Omega_m$, yielding
\begin{align} \label{dyn1}
    \frac{d x}{dN} &= x \Big(3x+y+1+3w(1-x-y) - \frac{4x}{x+2y}\Big)
    \, ,\\ \label{dyn2}
    \frac{d y}{dN} &= y \Big(3x+y-1+3w(1-x-y) - \frac{8x}{x+2y}
    \Big)
    \, .
\end{align}
The equations for $dx/dN$ and $dy/dN$ take a somewhat similar form, owing to the fact that the variables are related by $y=9 b_2 H^2 x^2/b_1^2$. One can also show that they satisfy the simple relation
\begin{equation}
    \frac{1}{x}\frac{dx}{dN} - \frac{1}{y}\frac{dy}{dN} = \frac{6 x + 4y}{x+2y} \, ,
\end{equation}
which follows from (\ref{dyn1})-(\ref{dyn2}), or from differentiating $y=9 b_2 H^2x^2/b_1^2$ and using (\ref{dyn1}) along with the expression for $\dot{H}$ (\ref{Hdot}). This equation is actually analytically solvable, and we refer to Appendix \ref{appendix:analyticalsolution} for exact solutions in this model.

An important property of the system is that the line $x=-2y$ displays divergent behaviour due to the denominator vanishing in the final terms of equations (\ref{dyn1})-(\ref{dyn2}). Similarly, the adiabatic sound speed diverges on this line. With the exception of the matter-dominated limit $x\rightarrow0$ and $y\rightarrow0$, which must be studied carefully, all other points on the $x=-2y$ line represent genuine pathological behaviour. However, the perturbative stability constraints will exclude these regions from the  physical phase space.

\subsection{Fixed point analysis and stability constraints}\label{sec:fixed}

In this part we analyse the \textit{fixed points} $(x_*,y_*)$ of the system and their linear stability, characterised by the eigenvalues of the Jacobian matrix \cite{perko2013differential}
\begin{equation} \label{stab}
J(x_{*},y_{*})= \left.\left( \begin{array}{cc}
\displaystyle\frac{\partial {f_1}}{\partial x}&     
\displaystyle\frac{\partial {f_1}}{\partial y}\\
\vspace{-8pt}\\
\displaystyle\frac{\partial {f_2}}{\partial x}&     
\displaystyle\frac{\partial {f_2}}{\partial y}
\end{array} \right)\right|_{(x,y)=(x_*,y_*)}\, ,\\ 
\end{equation}
where $f_1$ and $f_2$ are defined by equations (\ref{dyn1}) and (\ref{dyn2}) respectively. We will also include the perturbative stability constraint $c_s^2 \geq 0$ and positive energy condition $\rho_{\phi} \geq 0$ in our phase space analysis. Lastly, regions with accelerated expansion $-1 \leq q < 0$ and phantom expansion $q < -1$ will be highlighted. A combination of all of these constraints will determine the physically relevant phase space and in turn give constraints on the parameters $b_1$ and $b_2$.

The fixed points $(x_*,y_*)$ of the autonomous system (\ref{dyn1})-(\ref{dyn2}) are given in Table \ref{tab:fixed}. In what follows, we fix the matter equation of state to describe pressureless dark matter $w=0$.
Point A at $(0,1)$ is dominated by the quartic part of the scalar field $y= \Omega_{\phi} = 1$, with $\Omega_m=0$. The equation of state is $w_{\textrm{eff}}=w_{\phi}=1/3$ and the adiabatic sound speed squared is positive $c_s^2 = 1/3$. This is the radiation tracking solution \cite{Armendariz-Picon:2000nqq}, and its appearance in the quadratic kinetic $k$-essence models is well-studied \cite{Scherrer:2004au}.
The eigenvalues of the linear stability matrix (\ref{stab}) are $\lambda_1 = 1$, $\lambda_2=2$, indicating an unstable fixed point and a past attractor of the physical phase space. The second point B at $(1,0)$ is dominated by the quadratic part of the scalar field $x=1$. This corresponds to the usual stiff-matter kinetic point found in quintessence models \cite{Copeland:1997et}, with $w_{\textrm{eff}}=w_{\phi}=1$ and $c_s^2 =1$. It acts as a saddle, with eigenvalues $\lambda_1 = 3$, $\lambda_2=-6$, but is not relevant for late-time dynamics. The last scalar field dominated fixed point $\Omega_{\phi}=1$ is the de Sitter attractor C at $(-2,3)$. This represents an accelerated expansion solution with $w_{\textrm{eff}}=w_{\phi}=-1$ and is perturbatively stable $c_s^2 =0$. The eigenvalues are $\lambda_{1} = \lambda_{2} =-3$, and existence requires both $b_1$ and $b_2$ to be positive. It follows that only the completely symmetric geometries of Section \ref{sec:2} are compatible with these solutions. Hence, we see the importance of the new vectorial nonmetricity contribution in (\ref{vectorial_nonmetricity}).

Finally, we point out that the matter dominated `point' at $(0,0)$, which has the physical properties $\Omega_m=1$, $\Omega_{\phi}=0$ and $w_{\textrm{eff}}=0$, is not captured by our dynamical system. This is because the limit $x\rightarrow0$ and $y\rightarrow0$ is undefined\footnote{Explicitly, the limit $x\rightarrow0$ and $y\rightarrow0$ is not defined for our system as it lies on the divergent line $x=-2y$, hence O is not technically a critical point. Similarly, the previous analysis using the same variables \cite{Chakraborty:2019swx} correctly does not list the origin as a fixed point. Analyses that do study these points often obtain incorrect results (e.g., the stability analysis is inconsistent, with the signs of the eigenvalues depending on how one approaches the matter point). Using the alternative set of variables \cite{Quiros:2025fnn} does not directly solve the issue, as the $y=1$ line is divergent (see Fig.~1 of  \cite{Quiros:2025fnn}). The limits taken to obtain the values of $q$ and $w_{\textrm{eff}}$ therefore appear inconsistent on this line.}.
However, when following a heteroclinic orbit A $\rightarrow$ O $\rightarrow$ C that passes arbitrarily close to A (but does not reach it), the origin acts as a matter saddle. This is confirmed by the phase diagrams and the numerical solutions, and by taking the limit along the $y$-axis from point A to point O or  the limit along the line $y=-3x/2$ from point C to point O. Further comments on these problematic points are given in Appendix \ref{sec:comment}.

\begin{table}[htb!]
    \centering
    \begin{tabular}[t]{|c|c|c|c|c|c|c|c|}
    \hline
         Point & $\Omega_m$ & $x$ & $y$ & $w_{\rm{eff}}$ & $c_s^2$  & stability & conditions \\ \hline\hline 
        A & $0$ & $0$ & $1$ & $\frac{1}{3}$ & $\frac{1}{3}$ &  unstable & $b_2 > 0$
         \\
          B & $0$ & $1$ & $0$ & $1$ & $1$ &  saddle  & $b_1 < 0$
         \\
          C & $0$ & $-2$ & $3$ & $-1$ & $0$ &  stable & $b_1 >0$ \& $b_2 > 0$
         \\
         \hline
         O & $1$ & $0$ & $0$ & $w$ & -- &  -- & none
         \\
         \hline
    \end{tabular}
    \caption{Fixed points of the system, with matter density parameter $\Omega_m$, effective equation of state $w_{\rm{eff}}$ and scalar field adiabatic sound speed squared $c_s^2$. Linear stability and existence conditions have also been given in the last two columns. The origin O is listed separately, as the limit is not well-defined in our phase space. However, the physical properties $\Omega_m$ and $w_{\textrm{eff}}$ can be determined directly from the cosmological field equations.}
    \label{tab:fixed}
\end{table}

The phase portraits for the system are given in Fig.~\ref{fig:A}. The variables $x$ and $y$ are unbounded from above and below and so the space is non-compact. However, trajectories do not cross the $x=0$ and $y=0$ lines, evident from (\ref{dyn1}) and (\ref{dyn2}). It follows that the signs of $b_1$ and $b_2$ determine which distinct part of the phase space the model corresponds to for all time $N = \log a$.
For a positive matter density $\Omega_m \geq 0$, the Friedmann constraint (\ref{hub_const}) subjects the physical phase space to lie below the line $y=1-x$, indicated by the solid black line; the grey region above this line is therefore unphysical.
Accelerated expansion (blue shaded region) with $-1 \leq q < 0 $ and the phantom regime (green shaded region) with $q < -1$ are shown in Fig.~\ref{fig:A1}. The de Sitter critical point C lies on the border of these two regions with $q =-1$. 

The second figure \ref{fig:A2} shows the regions which violate the $k$-essence stability conditions $\rho_{\phi} < 0$ (orange region with orange-dashed border) and
$c_s^2<0$ (red region with red-dashed border). We have also indicated the superluminal regions $c_s^2>1$ in purple. It is interesting to observe that the only  superluminal region satisfying the stability constraints~(\ref{constr2}), with $b_1<0$ and $b_2<0$,
will necessarily run into gradient instabilities in the past. This can be seen by tracing trajectories backwards until they hit the divergent line $y=-x/2$. Note that $c_s^2<0$ occurs between the two lines $y=-3x/2$ and $y=-x/2$, with the latter line being exactly where the dynamical equations diverge. Orbits within the triangle with vertices O, A and C are therefore safe from this pathological behaviour for all $N$. (This is made clearer in Fig.~\ref{fig:prior} of the following section, which shows a zoomed-in phase diagram.)

\begin{figure}[htb!]
    \centering
    \begin{subfigure}[t]{0.48\textwidth} 
        \centering
        \includegraphics[width=0.95\linewidth]{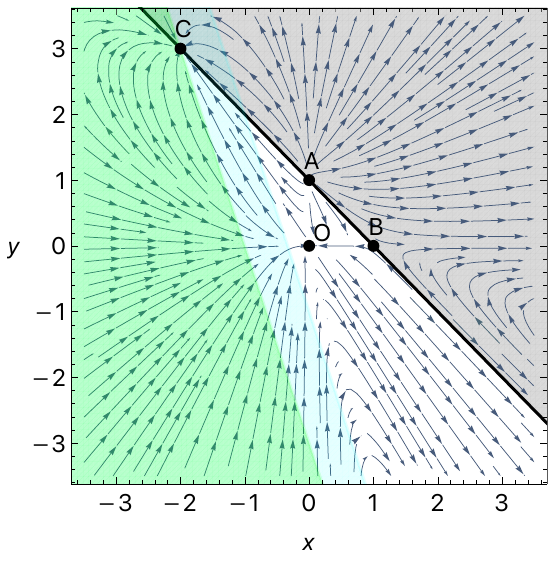}
        \caption{Phase space showing accelerated expansion $ -1 \leq w_{\textrm{eff}} \leq -1/3$ (blue) and phantom regime $w_{\textrm{eff}} <- 1$ (green).}\label{fig:A1}
    \end{subfigure}%
    ~ 
    \begin{subfigure}[t]{0.48\textwidth} 
        \centering
        \includegraphics[width=0.95\linewidth]{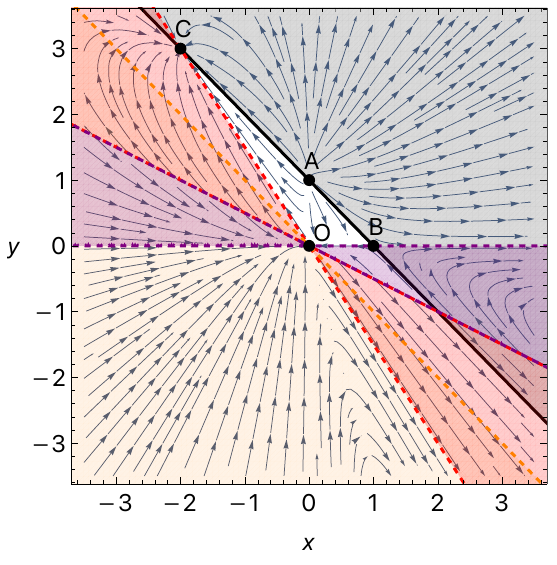}
        \caption{Phase space showing stability conditions. The red region indicates $c_s^2 < 0$, the orange region indicates $\rho_{\phi} < 0$, and the purple region  indicates $c_s^2 >1$.}\label{fig:A2}
    \end{subfigure}
    \caption{Phase portraits with pressureless matter equation of state $w=0$. The grey area (above the black line) is excluded due to a negative matter density $\Omega_m <0$. Points A, B and C lie on this line with $\Omega_m =0$. The triangle formed from the vertices O, A, C represents the physically relevant phase space, and orbits within this region stay in this region for all time.}  \label{fig:A}
\end{figure}

Realistic trajectories through the phase space will end at the fixed point C. The evolution of a trajectory following close to A $\rightarrow$ O $\rightarrow$ C is given in Fig.~\ref{fig:evol}.
Proximate to the O $\rightarrow$ C line defined by $y=-3x/2$, the scalar field behaves like a cosmological constant with $w_{\phi} \approx -1$ and sound speed $c_s^2 \approx 0$. For other nearby trajectories we have $w_{\phi} > -1$, with the phantom regime $w_{\phi}<-1$ strictly excluded for stability reasons ($c_s^2 <0$). We also note that this disconnected region of the phase space is free from superluminal sound speeds,      $c_s^2<1$. In fact, as previously pointed out, the physically acceptable regions with $b_1>0$, $c_s^2\geq0$ and $\rho_{\phi,X} >0$ immediately give 
a constrained propagation speed $0 \leq c_s^2<1/3$, see Eq. (\ref{constr}). We therefore avoid extended discussion of any subtleties related to superluminal propagation speeds, which we refer to \cite{Adams:2006sv,Bonvin:2006vc,Babichev:2007dw} for further details.

It should be stressed that the origin of our dark energy fluid is clearly distinct from a cosmological constant, with the early-time repeller for the physical phase space located at point A with $w_{\phi} =1/3$. Hence, both $\phi$ and $w_{\phi}$ will generically evolve at early times. Nonetheless, we can use the $\Lambda$CDM-like trajectory to find initial conditions and priors for our observational analysis. For positive $b_1$ and $b_2$, to closely track the line $y=-3x/2$ we have
\begin{equation} \label{b1b2}
    b_1 \lessapprox  b_2 \varphi^2 \, ,
\end{equation}
where the inequality comes from requiring it to be above the line (i.e., $c_s^2 \geq 0$ and $w_{\phi} \geq -1$). When the inequality is exact, the field acts like a cosmological constant with density $\rho_{\phi} = b_1^2/(4b_2)$. This is also exactly what happens for a quintessence field dominated by its potential, so the result is not surprising. What is novel about the result is that trajectories near the $\Lambda$CDM line remain close to that line from the matter-domination origin to the late-time accelerated expansion point C. This will turn out to be crucial for finding viable cosmological solutions that fit observational data.

\begin{figure}[htb!] 
    \centering
     \centering
        \includegraphics[width=0.6\linewidth]{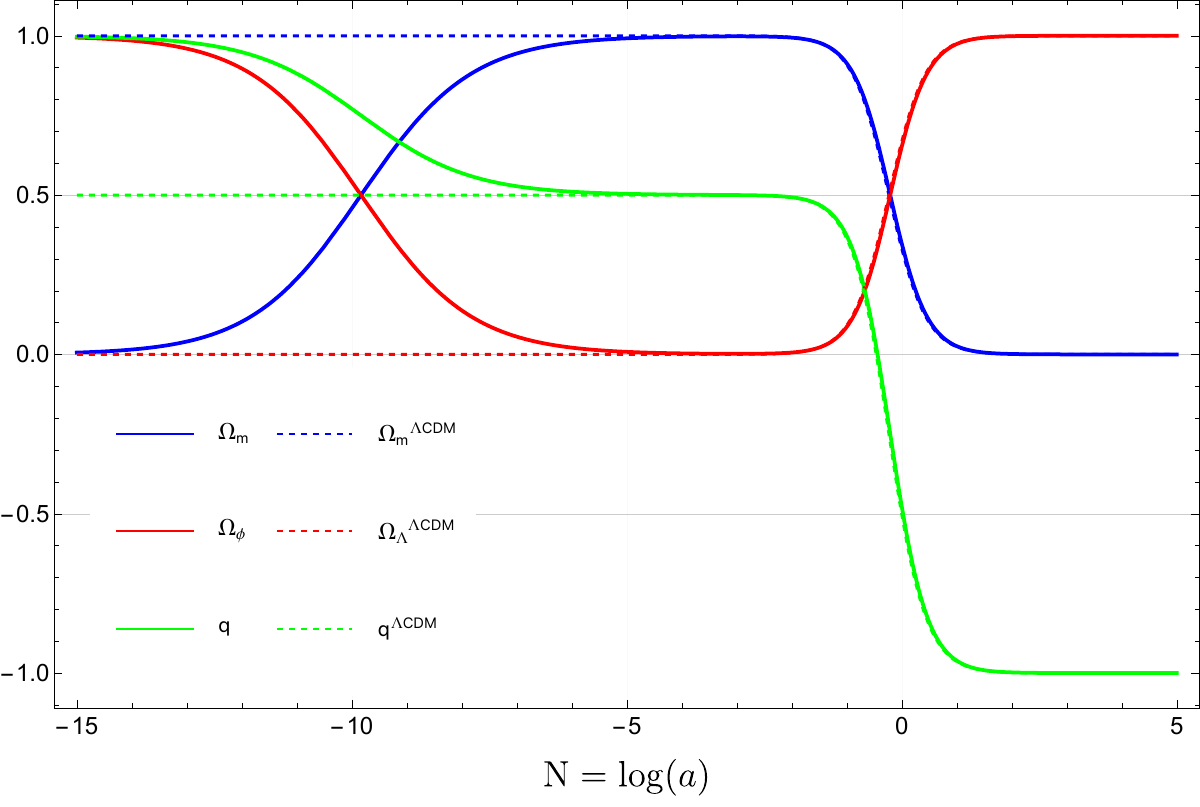}
    \caption{
Evolution of the matter density parameter $\Omega_m$, the $k$-essence density parameter $\Omega_\phi$, and the deceleration parameter $q$ in the geometric $k$-essence model, compared to the standard $\Lambda$CDM scenario.}
\label{fig:evol}
\end{figure}

The final piece of information required to give concrete values for $b_1$, $b_2$ and $\varphi_0$ is information about today's values of physical parameters. This will give a specific location along the trajectory close to O $\rightarrow$ C.
The easiest way to implement this is to fix a value for $\Omega_{m0}$. For instance, using $\Omega_{m0} \approx 0.32$ in the Friedmann constraint to locate our position along the line $y=-3x/2$ gives $x \approx -1.4$ from which we find $b_1 \approx 8.4 H^2/\varphi^2$ and $b_2 \approx 8.4 H^2/\varphi^4$. In Fig.~\ref{fig:evol} we have plotted a typical evolution using these initial conditions compared to the exact $\Lambda$CDM evolution. These typical values will be used to inform the priors for the MCMC runs in the next section.

Lastly, we observe that in order to have a sufficient period of matter domination, the initial matter density parameter $\Omega_{m0}$ cannot be too small (i.e., close to zero). If this were the case, trajectories would follow closely along the orbit A $\rightarrow$ C where $\Omega_{m} \approx 0$, and nowhere approach the origin O. The effective equation of state would then evolve smoothly between $w_{\textrm{eff}} = -1$ and $w_{\textrm{eff}} =1/3$, as opposed to displaying matter-dominated saddle behaviour (e.g., as shown in Fig.~\ref{fig:evol}). This would certainly be in conflict with observations, with no extended periods of time to allow for structure formation. In this respect, we agree with the conclusions of \cite{Quiros:2025fnn} that the purely kinetic quadratic $k$-essence model cannot adequately describe a unified dark sector, as initially envisioned by Scherrer \cite{Scherrer:2004au}. In this work, we therefore only aim to describe late-time acceleration and dark energy effects with our geometric $k$-essence model.

\section{Late-time observational constraints} \label{sec:5}
In this section, we constrain the parameters of the proposed geometric $k$-essence model using Markov Chain Monte Carlo (MCMC) techniques. For this purpose, we rewrite the cosmological equations into a dimensionless form in the redshift representation. The priors for the model parameters are guided by theoretical stability requirements, as detailed in section \ref{sec:4}. We then describe the observational datasets used in our analysis, along with the Python libraries employed in the MCMC sampling. Finally, we report the inferred constraints on the model parameters based on the combined observational data.

\subsection{Redshift representation}

Based on the dynamical systems analysis, the existence of de Sitter solutions requires  both $b_1>0$ and $b_2>0$. These signs can now be fixed, allowing us to make field redefinitions that reduce the number of free parameters in the model by one. 
We introduce a set of dimensionless variables $(h,\tau,r,\Phi,B)$ to implement these redefinitions as
\begin{equation} \label{new_vars}
h = \frac{H}{H_0} \, , \quad \tau =H_0 t\, , \quad r = \frac{\rho \kappa}{3 H_0^2} \, , \quad \Phi = \frac{\sqrt{b_1} \varphi}{H_0} \, , \quad B = \frac{3b_2 H_0^2}{4 b_1^2} \, .
\end{equation}
In these variables, the stability condition \eqref{b1b2} obtained from the stability analysis takes the form
\begin{equation} \label{Bconstr}
    B \geq \frac{3}{4 \Phi^2} \, ,
\end{equation}
while the Friedmann equations together with the evolution equation for the scalar field can be written as
\begin{align}
    3h^2 &=  3 r - \frac{1}{2} \Phi^2 +  B \Phi^4 \label{Friedmann1}   \, , \\
    3h^2 + 2 \frac{dh}{d\tau}&=  \frac{1}{2} \Phi^2 -  \frac{1}{3} B \Phi^4  \label{Friedmann2}\, , \\
    \frac{d \Phi}{d \tau} + 3h \Phi &= 4 B \Big(\frac{d \Phi}{d \tau} \Phi^2 + h \Phi^3 \Big) \label{Friedmann3} \, .
\end{align}
Here, we have assumed pressureless matter ($p=0$) corresponding to a dark matter or dust-dominated Universe. Equations \eqref{Friedmann1}-\eqref{Friedmann3} now take a more familiar form with the $\Phi^2$ term being canonically normalised, and $B$ being a dimensionless model parameter associated with the quartic nonmetricity contributions in (\ref{R_expand}). Introducing the redshift variable $z$ defined via $1+z=\frac{1}{a}$, the dimensionless time derivative transforms according to
\begin{equation}
    \frac{d}{d \tau}=-(1+z)h(z) \frac{d}{dz}.
\end{equation}
In terms of redshift, the cosmological equations describing the geometric $k$-essence model take the following form
\begin{align}  \label{ff1}
    3h(z)^2 &=3r(z)-\frac{1}{2} \Phi(z)^2+  B \Phi(z)^4 \, , \\
    \label{FF1}
    \frac{dh(z)}{dz}&=\frac{3h(z)^2-\frac{1}{2} \Phi(z)^2+  \frac{1}{3} B \Phi(z)^4}{2(1+z)h(z)} \, , \\
\label{Phi1}
    \frac{d\Phi(z)}{dz}&=\frac{-  4 B h(z) \Phi(z)^3 +3 h(z) \Phi(z)}{(1+z)h(z)(1-4 B \Phi^2)} \, .
\end{align}

We briefly comment on the different approaches for solving these equations: the first is to simply solve both equations numerically and use our solution for $h$ and $\Phi$ in the Friedmann constraint (\ref{ff1}) to determine the matter density $\Omega_m$. This is the most direct approach. The other option is to instead use the continuity equation to write the matter density parameter explicitly in terms of $\Omega_{m0}(1+z)^3$, and use the scalar field equation to get an integral expression for $\Phi$ to use in the Friedmann constraint (\ref{ff1}). This is commonly used in quintessence models with a free potential $V(\phi)$, where that freedom can be re-expressed in terms of some dark energy parametrisation $w_{\phi}(z)$ \cite{Linder:2002et,Corasaniti:2002vg}. Moreover, in general kinetic $k$-essence theories with arbitrary $P(X)$ this has also been shown to be possible \cite{Li:2006bx,Cardenas:2018ltl}.
This effectively trades a free model parameter for the matter density parameter today $\Omega_{m0}$. In fact, the scalar field equation \eqref{Phi1} is analytically solvable in terms of $a$ or equivalently $z$, see \cite{Scherrer:2004au} and Appendix \ref{appendix:analyticalsolution}  for further details. However, the stability constraints \eqref{Bconstr} which then relate $B$ and $\Phi_0$ take a more cumbersome form, so we instead solve the differential equations numerically. We later verify that both approaches are equivalent, up to the same degree of precision. 

To make the comparison with $\Lambda$CDM obvious, we evaluate the Friedmann constraint at redshift $z=0$, allowing us to trade the initial scalar field value $\Phi_0$ for the matter density parameter $\Omega_{m0}$. Imposing $h(0)=1$, $\Phi(0)=\Phi_0$, and $r(0)=\Omega_{m0}$ in Eq.~\eqref{ff1}, we obtain the relation
\begin{equation} \label{ff2}
    1 = \Omega_{m0} - \frac{1}{6} \Phi_0^2 + \frac{B}{3} \Phi_0^4 \, ,
\end{equation}
which can be solved for $\Phi_0$ in terms of $\Omega_{m0}$ and $B$. This lets us work with the more physically meaningful $\Omega_{m0}$ and directly compare our results to $\Lambda$CDM. The constraint equation (\ref{ff2}) can then be used to obtain $\Phi_0$ if desired. 

From this point onward, we work with $(\Omega_{m0},B)$.  Combining the energy conditions $\Omega_{m0} \geq 0$ and $\rho_{\phi} \geq 0$ with the perturbative stability condition (\ref{Bconstr}) leads to the following constraints on the model parameters
\begin{equation} \label{model_constr1}
    0 \leq \Omega_{m0} < 1 \, , \quad B \geq \frac{1}{16-16\Omega_{m0}} \, .
\end{equation}
These conditions define the physical parameter space, which corresponds to the region enclosed by the triangle OAC in Fig.~\ref{fig:prior}. Since $B$ can, in principle, grow arbitrarily large, the parameter space is unbounded. To restrict this, we additionally require that the Universe is accelerating today, i.e. the deceleration parameter \eqref{decel} satisfies $q_{0} < 0$. Using the variable definitions \eqref{new_vars} and the Hubble constraint \eqref{ff2}, we arrive at the final allowed parameter region
\begin{equation} \label{param_region}
    0 \leq \Omega_{m0} < 2/3 \, , \quad \frac{1}{16-16 \Omega_{m0}} \leq B < \frac{4-3\Omega_{m0}}{6(\Omega_{m0}-2)^2} \, ,
\end{equation}
which defines the physically viable and accelerating region of the model, as illustrated in Fig.~\ref{fig:prior}. These bounds will be used in our MCMC analysis to ensure sampling is restricted to the physically relevant regions of the phase space of the model.

\begin{figure}[htb!] 
     \centering   \includegraphics[width=0.5\linewidth]{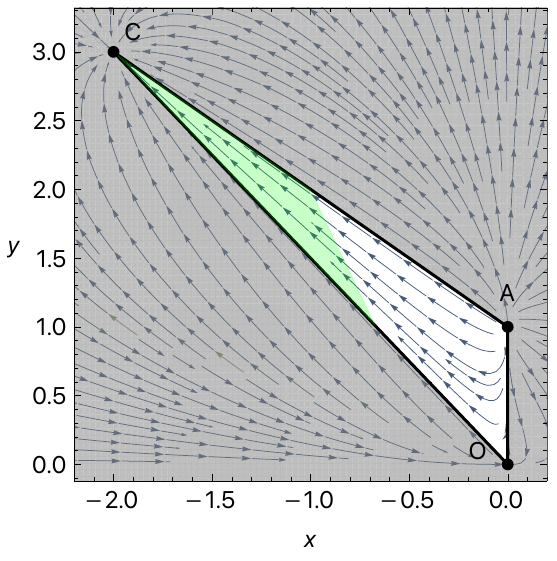}
    \caption{Phase plot showing the physical phase space (black border) and accelerating region (green)
    used to define the parameter space, see Eq.~(\ref{param_region}).}    \label{fig:prior}
\end{figure}

\subsection{Methodology and datasets}
To constrain the model parameters, we use a Markov Chain Monte Carlo (MCMC) method implemented with the \texttt{emcee} sampler \cite{ForemanMackey2013}. MCMC methods are commonly used in cosmology to constrain model parameters by exploring their posterior probability distributions, especially when the parameter space is high-dimensional. They work by building a Markov chain, whose stationary distribution matches the posterior. Unlike grid-based methods, MCMC can handle complex models more easily and avoids scanning the entire parameter space point by point \cite{Gregory2005}. The affine-invariant ensemble sampler we use is well-suited for cases where the posterior has complicated shapes or when parameters vary over very different scales. The posterior distribution being sampled is related to the likelihood and prior through Bayes' theorem \cite{Lewis2002}
\begin{equation}
P(\theta|D) = \frac{L(D|\theta) P(\theta)}{P(D)}  \, ,
\end{equation}
where $P(\theta|D)$ is the posterior probability of the model parameters $\theta$ given the data $D$. The likelihood $L(D|\theta)$ quantifies the agreement between the model and the data, while the prior $P(\theta)$ encodes any existing knowledge or assumptions about the model parameters. The evidence $P(D)$ serves as a normalisation factor ensuring the posterior integrates to unity, though it is typically not required in parameter estimation.

A key advantage of MCMC is that it not only determines the most probable values of the model parameters, but it also naturally incorporates uncertainties arising both from the model and the observational data \cite{Gelman2013}. As outlined above, applying MCMC in our case requires two essential ingredients:
\begin{enumerate}
    \item The likelihood function, which quantifies the agreement between the model's theoretical predictions and the observational data;
    \item The prior distributions, which encode assumptions or constraints on the model parameters and are supplied to the \texttt{emcee} sampler.
\end{enumerate}

The form of the likelihood function depends strongly on the specific datasets employed. In this study, we make use of three datasets, detailed below. The explicit form of the likelihoods corresponding to each dataset is presented following their description.
\begin{itemize}
    \item[$\triangleright$] \textit{Cosmic Chronometers:} We use Hubble measurements extracted based on the differential age method \cite{Jimenez_2002}. This technique makes use of passively evolving massive galaxies that formed at redshifts $z \sim 2$-$3$, allowing for a direct and model-independent determination of the Hubble parameter using the relationship $\Delta z/ \Delta t$. The differential age method significantly reduces the dependence on the underlying astrophysical assumptions, allowing for model-independence. In our approach, there are two main differences compared to the parameter inference performed in \cite{Vagnozzi_2021}. First, we use only a subset of the $31$ points used in \cite{Vagnozzi_2021}, selecting 15 points compiled in \cite{moresco2020HzTable} from \cite{Moresco_2012, 10.1093/mnrasl/slv037, Moresco_2016}, i.e. those with a full covariance matrix available. Secondly, for analysing these points, we use the full non-diagonal covariance matrix, as opposed to the diagonal covariance matrix used in \cite{Vagnozzi_2021}. 
    \item[$\triangleright$] \textit{Type Ia supernovae:} We also use the Pantheon$^{+}$ dataset, which includes $1701$ light curves from $1550$ Type Ia supernovae (SNe Ia) spanning the redshift range $0.001 \leq z \leq 2.26$ \cite{Brout_2022}. A publicly available total covariance matrix, incorporating both statistical and systematic uncertainties, is provided with this dataset \cite{PantheonPlusGitHub, Conley2011}, and will be discussed later in the context of the likelihood construction.
    \item[$\triangleright$] \textit{Baryon Acoustic Oscillations:} Lastly, we incorporate measurements of baryon acoustic oscillations (BAO), which provide a standard ruler for cosmic distances based on the imprint of sound waves in the early universe. These measurements constrain the cosmic expansion history by probing both the comoving angular diameter distance $D_{M}(z)/r_{d}$ and the Hubble distance $D_{H}(z)/r_{d}$ across a range of redshifts. We use the compilation of BAO data from the DESI DR1 release \cite{Adame2025}, which includes multiple galaxy and quasar samples. The data span an effective redshift range from $z_{eff}=0.295$ to $z_{eff}=2.33$ with associated measurements of $D_{M}/r_{d}, D_{H}/r_{d}$ or $D_{V}/r_{d}$, and their uncertainties, as summarised in Table 1 of \cite{Adame2025}. These BAO measurements are especially valuable when combined with SNe Ia and CC data, as they help break degeneracies in the estimation of cosmological parameters.
\end{itemize}
We now provide the two key components needed for the MCMC analysis. The total likelihood is built as the sum of three contributions, each associated with one of the datasets introduced earlier through the formula
\begin{equation}\label{likelihood}
    \log \mathcal{L}_{\text{tot}}=-\frac{1}{2} \chi^2_{\text{tot}}+\text{const} \, ,
\end{equation}
where the constant term arises from normalisation. Since it does not depend on the model parameters, it is irrelevant for parameter inference and can be omitted. The total chi-squared appearing in Eq.~\eqref{likelihood} is given by
\begin{equation}\label{chi2tot}
\chi^2_{\text{tot}}=\chi^{2}_{\text{CC}}+\chi^{2}_{\text{SNeIa}}+\chi^{2}_{\text{BAO}}\, , \; \; \text{where}
\end{equation}
\begin{itemize}
    \item[$\triangleright$] $\chi^{2}_{\text{CC}}$ represents the contribution from Cosmic Chronometers (CC), computed using the covariance matrix developed by Moresco \cite{Moresco2021CCcovariance}. It is defined as
    \begin{equation}
        \chi^{2}_{\text{CC}}(\theta,z)=\Delta H^{T}(z) \mathbf{C}^{-1}_{\text{tot}} \Delta H(z) \, ,
    \end{equation}
   where $\Delta H(z)=H_{\text{model}}(\theta,z)-H_{\text{obs}}(z)$ is the vector of residuals between the model prediction and observed values of the Hubble parameter, and $\theta$ denotes the model parameters. The total covariance matrix $\mathbf{C}_{\text{tot}}$ incorporates both statistical and systematic contributions
   \begin{equation}
       \textbf{C}_{\text{tot}}= \textbf{C}_{\text{stat}}+\textbf{C}_{\text{sys}} \, ,
   \end{equation}
   where $\textbf{C}_{\text{stat}}$ is a diagonal matrix containing the statistical uncertainties of the $H(z)$ measurements, and $\textbf{C}_{\text{sys}}$ accounts for the correlated systematic errors.
    \item[$\triangleright$] $\chi^{2}_{\text{SNeIa}}$ corresponds to the contribution from the Pantheon$^{+}$ supernova sample, and is defined as
    \begin{equation}
        \chi^2_{\text{SNeIa}}(\theta,z)=\Delta \textbf{D}^{T}(\theta,z) \textbf{C}^{-1}_{\text{tot}} \Delta \textbf{D}(\theta,z) \, ,
    \end{equation}
    where $\textbf{C}^{-1}_{\text{tot}}$ is the inverse of the full covariance matrix provided with the Pantheon$^{+}$ data release \cite{PantheonPlusGitHub}, and $\Delta \textbf{D}$ is the vector of residuals for the $1701$ supernovae in the sample. These residuals are computed as\footnote{Note the sign convention: for the Pantheon$^{+}$ dataset, the residual is defined as data minus model, unlike the model minus data convention often used for Hubble parameter measurements. This is purely a matter of convention and has no impact on the resulting likelihood, since the chi-squared is quadratic in the residual: $-\Delta \textbf{D}^{T} \textbf{C}^{-1} (-\Delta \textbf{D})=\Delta \textbf{D}^{T} \textbf{C}^{-1} \Delta \textbf{D}$.}
    \begin{equation}
    \Delta \textbf{D}(\theta,z)=\mu(z)-\mu_{\text{model}}(\theta,z) \, ,
    \end{equation}
   where the observed distance modulus $\mu(z)$ is related to the apparent magnitude $m(z)$ by
    \begin{equation}
        \mu(z)=m(z)-\mathcal{M} \, ,
    \end{equation}
    with $\mathcal{M}$ denoting the absolute magnitude of a standard Type Ia supernova. The apparent magnitude $m(z)$ is given in terms of the luminosity distance $d_{L}(z)$ as
    \begin{equation}
        m(z)=5 \log_{10} \left(\frac{d_{L}(z)}{Mpc} \right) + \mathcal{M}+25 \, ,
    \end{equation}
   where the luminosity distance itself is determined from the Hubble parameter via
    \begin{equation}
        d_{L}(z)=c(1+z) \int_{0}^{z} \frac{ dz'}{H(z')} \, ,
    \end{equation}
    with $c$ being the speed of light in $km/s$.
    \item[$\triangleright$] $\chi^{2}_{\text{BAO}}$ accounts for the contributions from the DESI DR1 BAO measurements. To incorporate these contributions, we use the following distance measures, from which $\chi^2_{\text{BAO}}$ is constructed:
    \begin{itemize}
        \item \textit{radial Hubble distance}
        \begin{equation}
            D_{H}(z)=\frac{c}{H(z)} \, ,
        \end{equation}
        \item \textit{comoving angular diameter distance} \begin{equation}
            D_{M}(z)=c \int_0^{z} \frac{ dz'}{H(z')} \, ,
        \end{equation}
        \item \textit{volume averaged distance}
        \begin{equation}
            D_{V}(z)=\left(z D_M^2(z) D_{H}(z) \right)^{\frac{1}{3}} \, .
        \end{equation}
    \end{itemize}
    Each of these distances is normalised by the sound horizon at the drag epoch, $r_{d}$, as provided in the DESI DR1 data \cite{Adame2025} (see Table 1 therein). The contribution of the BAO data to the total chi-squared is computed using the full covariance matrix and takes the form
\begin{equation}
\chi^2_{\text{BAO}}(\theta,z) = 
\left[ \mathbf{D}_{\text{model}}(\theta,z) - \mathbf{D}_{\text{data}}(z) \right]^{T} 
\, \mathbf{C}^{-1} \,
\left[ \mathbf{D}_{\text{model}}(\theta,z) - \mathbf{D}_{\text{data}}(z) \right] \, ,
\end{equation}
where \( \mathbf{D}_{\mathrm{data}} \) is a vector containing the DESI measurements of \( D_M(z)/r_d \), \( D_H(z)/r_d \), and \( D_V(z)/r_d \), ordered to match the structure of the published covariance matrix \( \mathbf{C} \) \cite{desi2024_bao_covariance}. The theoretical predictions \( \mathbf{D}_{\mathrm{model}}(\theta,z) \) are obtained by evaluating the appropriate ratios, once $H(z)$ is known. The observational data are thus provided in terms of ratios $D_{X}/r_d$ for $X \in \{H,M,V\}$. In our approach, we treat $r_{d}$ as a free parameter rather than fixing it to the fiducial value inferred by the Planck collaboration\footnote{We note that allowing $r_d$ to freely vary leads to posterior values that remain close to the Planck-inferred fiducial value (see Sec. \ref{sec:data:results}), thus validating this approach.}. This choice improves model-independence by avoiding assumptions about the early Universe and recombination physics \cite{universe9090393, Lin_2021, Jedamzik2021}.

\end{itemize}

From Eq.~\eqref{chi2tot}, it is evident that $\chi^2_{\text{tot}}$ is fully determined by the observational data and the model-dependent Hubble parameter $H_{\text{model}}(\theta,z)$. In the case of our geometric $k$-essence model, the parameter set is $\theta=\{H_{0},\Omega_{m0},r_{d},\mathcal{M},B\}$. $H_{\text{model}}(\theta,z)$ is obtained numerically as the solution\footnote{We employ the BDF integration method from the \texttt{solve\_ivp} package, using relative and absolute tolerances of $10^{-5}$ and $10^{-15}$, respectively.} of the differential equations given in  \eqref{FF1}-\eqref{Phi1}, subject to the initial conditions
\begin{equation} h(0) = 1, \quad \Phi_0 = \frac{ \sqrt{1 + \sqrt{1 + 48B(1-\Omega_{m0})}}}{2\sqrt{B}} \, ,
\end{equation}
where $\Phi_{0}$ corresponds to the real-valued solution of Eq.~\eqref{ff2} that satisfies the stability conditions (up to an overall irrelevant sign). We have also verified that these numerical solutions coincide with the analytic solution given in Eq.~(\ref{hfull}) of Appendix \ref{appendix:analyticalsolution} for all $z$, with an error less than $10^{-14}$.

The final step in setting up the MCMC analysis is to choose appropriate priors for these parameters. While $H_{0}$, $\Omega_{m0}$, $r_{d}$ and $\mathcal{M}$ are commonly constrained in cosmological analyses, the $k$-essence specific $B$ parameter requires special care. In the case of the quadratic ghost condensate model, this is likely due to the need for carefully chosen priors to avoid pathological behaviour. In particular, the stability conditions outlined in Section \ref{sec:4} must be strictly enforced. For example, if the squared sound speed $c_s^2$ becomes negative, the model becomes physically unstable, which leads to divergent background evolution and numerical breakdowns. This behaviour is also seen in the phase space analysis, where trajectories entering the $c_s^2 <0$ region typically diverge. The energy conditions $\Omega_{m} \geq 0, \rho_{\phi} \geq 0$, together with the conditions for acceleration $q <0$ and  perturbative stability $c_{s}^2 >0$, $\rho_{\phi,X} >  0$ restrict the $k$-essence parameter $B$ to a bounded admissible range, given by Eq.~\eqref{param_region}.
This prior range ensures both physical viability and numerical stability throughout the MCMC sampling. Accordingly, we impose the uniform (flat) priors over the intervals
\begin{equation}
    50 \leq H_0 \leq 100 \, , \; \; 0 \leq \Omega_{m0} \leq 0.5 \, , \; \; 100 \leq r_{d} \leq 300 \, , \; \;  -20 \leq \mathcal{M} \leq -18 \, ,
\end{equation}
for the standard parameters, and the following uniform priors for the stability-dependent parameter
\begin{equation} \label{B_prior}
    \frac{1}{16-16 \Omega_{m0}} \leq B < \frac{4-3\Omega_{m0}}{6(\Omega_{m0}-2)^2} \, .
\end{equation}

Given these priors and the likelihoods described earlier, we initialise the MCMC sampling with a uniform distribution over the domain
\begin{equation} \label{priors}
    [50,100] \times [0,0.5] \times [100,300] \times [-20,-18] \times  \left[\frac{1}{16-16 \Omega_{m0}}, \frac{4-3\Omega_{m0}}{6(\Omega_{m0}-2)^2} \right] \, ,
\end{equation}
and employ the \texttt{emcee} ensemble sampler with $50$ walkers evolved over $20000$ steps. We discard the initial $10 \%$ (i.e., $2000$ steps) as burn-in and apply a thinning factor of \texttt{thin}=$30$ to reduce autocorrelation in the sampled chains.

To evaluate the statistical quality of our model fits, we employ the following statistical measures:
\begin{itemize}
    \item[$\triangleright$] \textit{Reduced $\chi^{2}$ statistic} \cite{Tauscher_2018}
    \begin{equation}        \chi^2_{\text{red}}=\frac{\chi^2_{\text{min}}}{N-k} \, ,
    \end{equation}
    where $\chi^2_{\text{min}}$ is the minimum chi-squared value, $N$ is the total number of data points, and $k=|\theta|$ is the number of free model parameters. In our case, $N=1701+15+12=1728$\footnote{When we include the R22 prior, which is one point, $N=1729$.}, corresponding to the Pantheon$^{+}$, CC and BAO datasets, while $k=5$ for the parameter set $\theta=\{H_0,\Omega_{m0}, B,r_d,\mathcal{M}\}$.
    \item[$\triangleright$] \textit{Akaike information criterion} \cite{10.1111/j.1745-3933.2007.00306.x, 10.1111/j.1365-2966.2011.19969.x}
    \begin{equation}\label{AICdef}
        \mathrm{AIC} = \chi^2_{\text{min}} + 2k \, ,
    \end{equation}
    where the first term quantifies the fit quality and the second imposes a penalty proportional to the number of parameters to discourage overfitting.
    \item[$\triangleright$] \textit{Bayesian information criterion} \cite{Burnham2002}
    \begin{equation}\label{BICdef} 
        \mathrm{BIC} = \chi^2_{\text{min}} + k \log N \, ,
    \end{equation}
    where $\log N$ increases the penalty strength with dataset size, making BIC more conservative in selecting models as $N$ grows.
\end{itemize}

It is important to note that each statistical measure serves a distinct purpose and comes with its own advantages and limitations. Specifically, the reduced chi-squared statistic, $\chi^2_{\text{red}}$, evaluates the goodness of fit between a model and the data, without penalising model complexity. In contrast, the Akaike Information Criterion (AIC) and the Bayesian Information Criterion (BIC) are model selection tools that balance goodness of fit with model complexity by including a penalty term for the number of parameters \cite{Vrieze2012, Rezaei2021, Arevalo2017}. Among the two, BIC imposes a stronger penalty than AIC for large datasets $\left( \log N > 2\right)$, as evident from their definitions. Since the geometric $k$-essence model has one additional free parameter compared to $\Lambda$CDM, AIC and BIC penalise it differently. The BIC penalty is stronger because $N$ is large, and therefore $\Delta{\rm BIC}$ is larger than $\Delta{\rm AIC}$.

\subsection{Results}\label{sec:data:results}
Posterior constraints extracted from the MCMC chains are listed in Table \ref{optimalparameters}. The corresponding parameter constraints are illustrated in the contour plots of Fig.~\ref{contourplots}. Notably, the $k$-essence specific parameter $B$  is constrained in a manner closely resembling that of $\Omega_{m0}$ in the $\Lambda$CDM model. This behaviour is explained in more detail in Appendix \ref{appendix:analyticalsolution}, where the full analytical solution of the geometric $k$-essence model is presented. The stability constraint implemented in the priors~(\ref{B_prior}) can also be clearly seen in the $\Omega_{m0}$-$B$ contour plots, with a sharp cutoff preventing $\Omega_{m0}$ from taking larger values while $B$ remains small.
The optimal values for $H_0, \mathcal{M}$ and $r_{d}$ extracted from the MCMC chains are almost identical to those of $\Lambda$CDM, while the matter density parameter $\Omega_{m0}$ favours lower values in the $k$-essence model. This is because the $k$-essence fluid behaves similar to DM at the background level at intermediate redshifts. The optimised values of $H_0$ and $\Omega_{m0}$ are also similar to those found in the recent study of the purely kinetic DBI $k$-essence models, which used similar observational data \cite{Ganguly:2025kdf}.
Together with the contour plots, these results validate the close correspondence with late-time $\Lambda$CDM dynamics.

\begin{table}[htb]
\centering
\begin{tabular}{|c|c|c|c|c|}
\hline
{Cosmological Model} & {Parameter} & {Prior} & {JOINT} & {JOINT \textbf{+} R22} \\
\hline
\multirow{4}{*}{{\(\Lambda\)CDM}} 
& \(H_0\) & [50, 100] & \(67.0 \pm 4.0\) & \(72.6 \pm 1.0\) \\
& \(\Omega_{m0}\) & [0, 0.5] & \(0.332 \pm 0.013\) & \(0.329 \pm 0.013\) \\
& \(\mathcal{M}\) & [–20, –18] & \(-19.45 \pm 0.13\) & \(-19.269 \pm 0.030\) \\
& \(r_d\) & [100, 300] & \(148.5^{+7.6}_{-9.5}\) & \(136.8 \pm 2.3\) \\
\hline
\multirow{5}{*}{{Geometric \(k\)-essence}} 
& \(H_0\) & [50, 100] & \(67.0 \pm 4.0 \) & \(72.6 \pm 1.0\) \\
& \( \Omega_{m0}\) & [0, 0.5] & \(0.291^{+0.042}_{-0.022}\) & \(0.288^{+0.042}_{-0.021}\)\\
& \(B\) & \( \left[\frac{1}{16-16 \Omega_{m0}}, \frac{4 - 3 \Omega_{m0}}{6(\Omega_{m0}-2)^2} \right] \) & \( 0.0932\pm 0.0018\) & \(0.0929 \pm 0.0018\) \\
& \(\mathcal{M}\) & [–20, –18] & \(-19.45^{+0.13}_{-0.12}\) & \(-19.270 \pm 0.030\) \\
& \(r_d\) & [100, 300] & \(148.6^{+7.7}_{-9.6}\) & \(136.8 \pm 2.3\) \\
\hline
\end{tabular}
\caption{Summary of maximum-likelihood values extracted from the MCMC chains, with \(2\sigma\) uncertainties, for the cosmological parameters in the {\(\Lambda\)CDM} and {geometric \(k\)-essence} model, derived from the JOINT and JOINT {+} R22 datasets.}\label{optimalparameters}
\end{table}

\begin{figure}[htb] 
\begin{subfigure}{.49\textwidth}
\includegraphics[width=\linewidth]{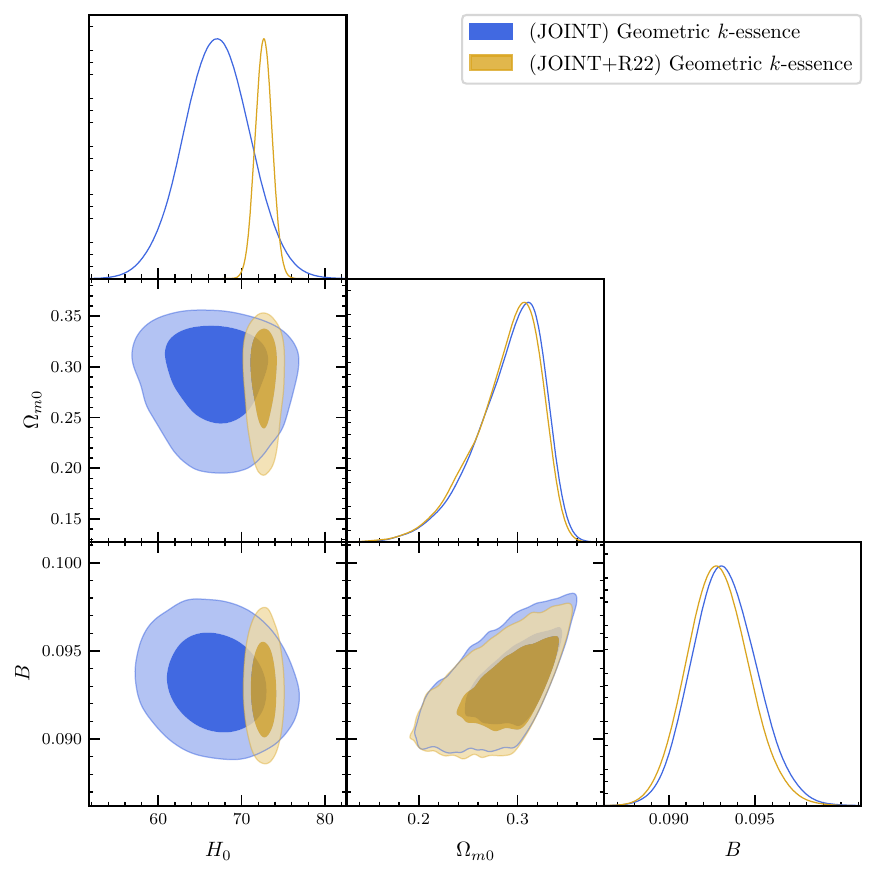}
\end{subfigure}
\hfil
\begin{subfigure}{.49\textwidth}
\includegraphics[width=\linewidth]{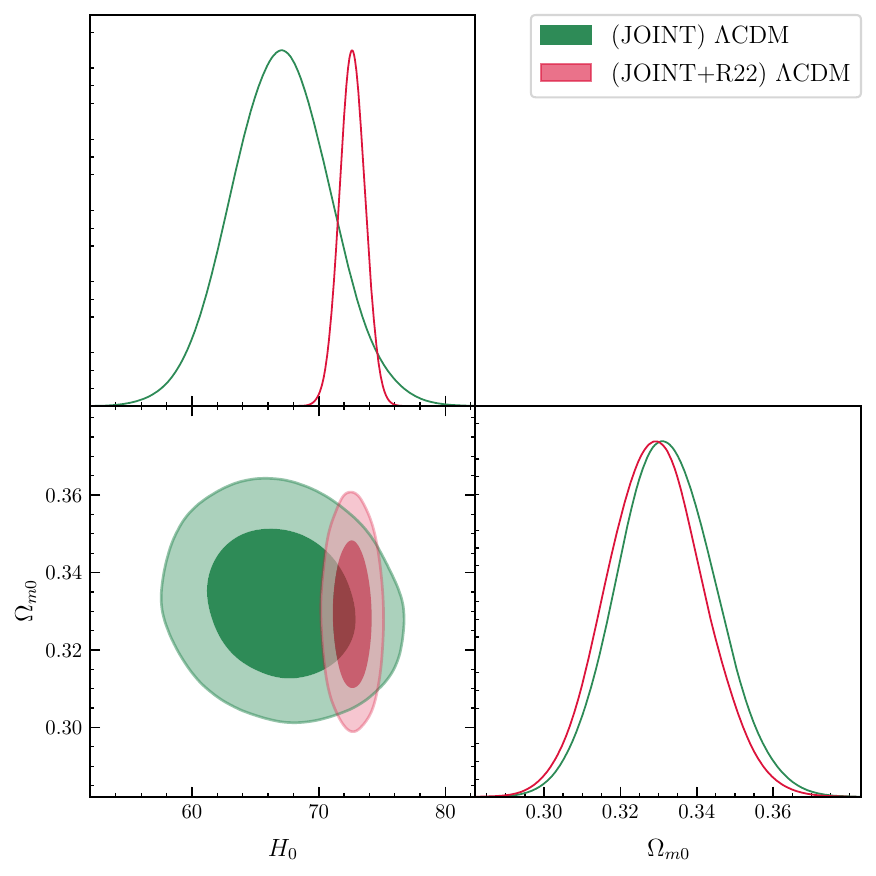}
\end{subfigure}
\caption{Corner plots showing the posterior constraints on the model parameters for the geometric $k$-essence model (left panel) and the $\Lambda$CDM model (right panel). Both the $1\sigma$ and $2\sigma$ confidence regions are displayed.}\label{contourplots}
\end{figure}

The statistical results represented in Table \ref{statisticalresults} show that the geometric $k$-essence model fits the data points well, with a reduced chi-squared close to unity ($\chi^2_{red} \simeq 1$). Despite the similar goodness of fit, the $k$-essence model has an additional parameter compared to $\Lambda$CDM, leading to a slightly worse AIC and BIC, with the latter being more significantly impacted by $P_{\textrm{tot}}$. Consequently, the model is comparable to $\Lambda$CDM with respect to AIC, but disfavoured by BIC. In the future, it would be interesting to conduct a more comprehensive comparison of competing purely kinetic $k$-essence models (such as the DBI models \cite{Ganguly:2025kdf}) to assess their relative viability.

\begin{table}[htb]
\centering
\begin{tabular}{|l|c|c|c|c|c|c|c|c|}
\hline
{Model} & \({\chi^2_{\text{tot,min}}}\) & \({P_{\text{tot}}}\) & \({N_{\text{tot}}}\) & \({\chi^2_{{red}}}\) & {AIC} & \({\Delta\text{AIC}}\) & {BIC} & \({\Delta\text{BIC}}\) \\
\hline
{\(\Lambda\)CDM} & 1778.14 & 4 & 1728 & 1.031 
 & 1786.14 & 0 & 1807.96 & 0 \\
{Geometric \(k\)-essence} & 1778.16  & 5 & 1728 & 1.032
 & 1788.16  & 2.02  & 1815.44 &  7.48
 \\
\hline
\end{tabular}

\caption{Model comparison based on total minimum chi-squared, number of free parameters (\(P_{\text{tot}}\)), number of data points (\(N_{\text{tot}}\)), reduced chi-squared, Akaike Information Criterion (AIC), and Bayesian Information Criterion (BIC). Values are given relative to the \(\Lambda\)CDM model.}
\label{statisticalresults}
\end{table}

\begin{figure}[htb] 
\begin{subfigure}{.48\textwidth}
\includegraphics[width=\linewidth]{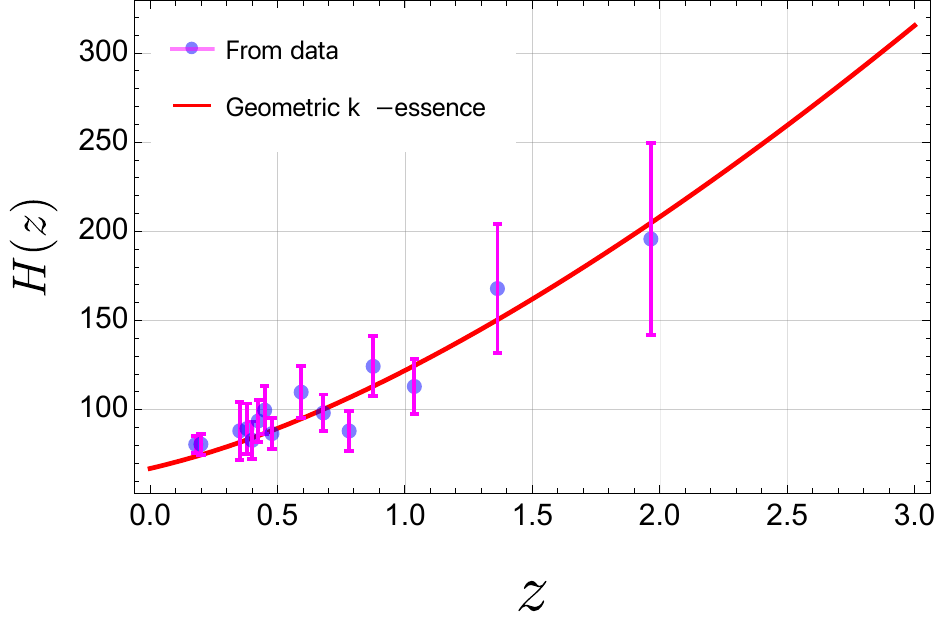}
\end{subfigure}
\hfil
\begin{subfigure}{.48\textwidth}
\includegraphics[width=\linewidth]{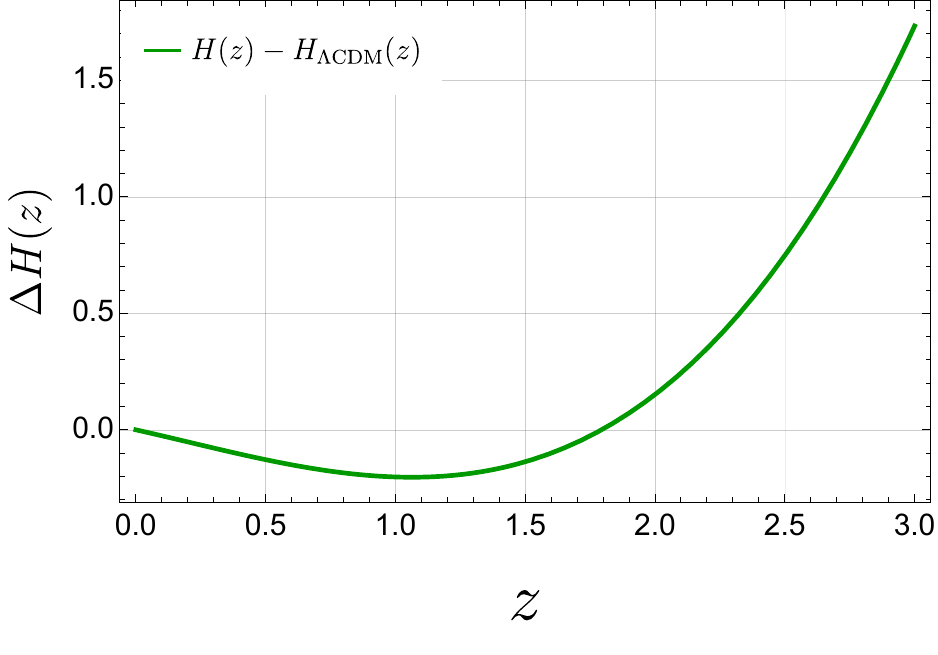}
\end{subfigure}
\caption{Plot of the Hubble function with Moresco CC data points (left panel) 
and the Hubble function difference (right panel) between the $k$-essence model and $\Lambda$CDM.
}\label{fig:hub}
\end{figure}

\begin{figure}[htb] 
\begin{subfigure}{.48\textwidth}
\includegraphics[width=\linewidth]{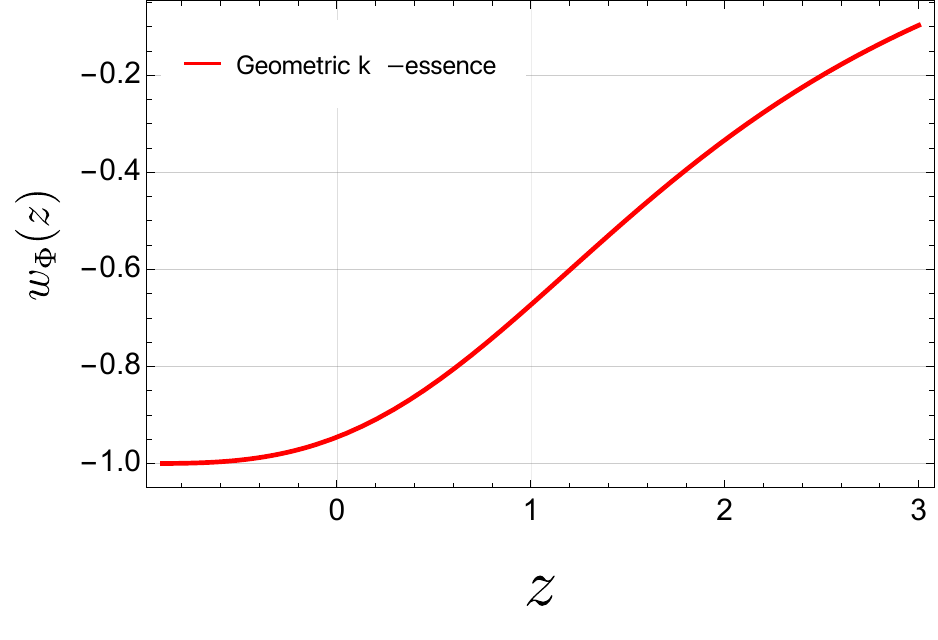}
\end{subfigure}
\hfil
\begin{subfigure}{.48\textwidth}
\includegraphics[width=\linewidth]{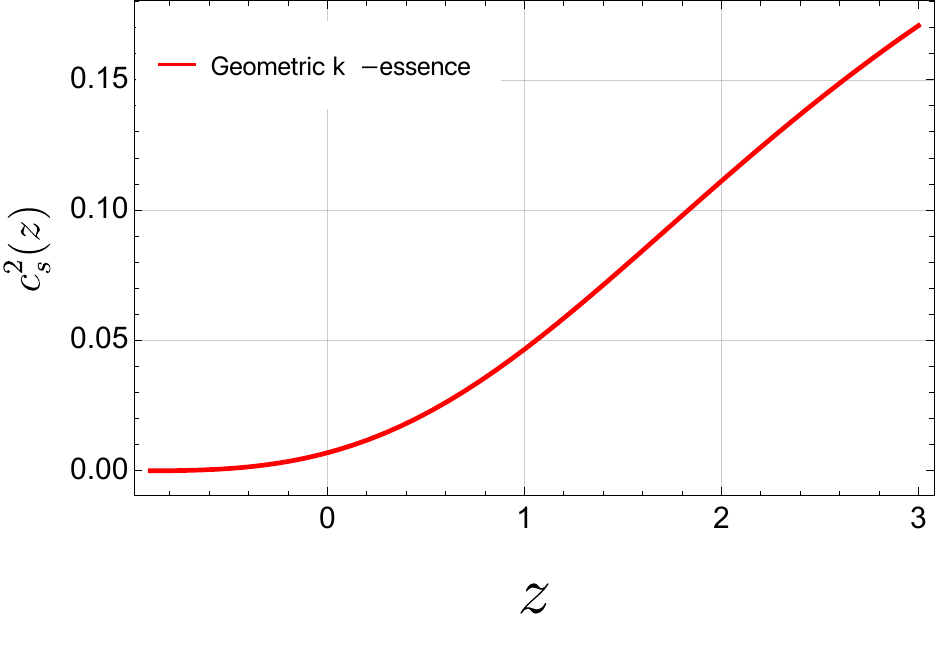}
\end{subfigure}
\caption{Redshift evolution of the equation of state parameter $w_{\Phi}(z)$ (left panel) and the sound speed squared $c_s^2(z)$ (right panel) for the geometric $k$-essence model.}\label{fig:wcs}
\end{figure}

The goodness of fit discussed earlier is visually confirmed in Fig.~\ref{fig:hub}. At late times, the Hubble functions of the $k$-essence and $\Lambda$CDM models are nearly indistinguishable, with differences remaining below $2$ (less than 1\%) up to redshift $z = 3$. In Fig.~\ref{fig:wcs} we show the evolution of the $k$-essence fluid parameters. The equation of state parameter $w_{\Phi}$ exhibits evolving dynamical behaviour, and the sound speed squared $c_s^2$ remains positive and stable throughout the evolution. 
Notably, $c_s^2$ stays within the range $0 < c_s^2 < 1/3$, with radiation-like behaviour at early times corresponding to the tracking-solution fixed point A of Table \ref{tab:fixed}. These plots show that while $w_{\Phi}$ clearly does not behave like a cosmological constant at low redshifts, the total equation of state (and hence the observables) still closely track $\Lambda$CDM at late times. This is a direct consequence of the $k$-essence fluid also generating effective dark matter contributions.
The plots also hint that
the $k$-essence fluid could have non-trivial effects on matter clustering in the intermediate-to-late Universe. The effects of this small late-time dark energy clustering could be of phenomenological interest, such as through suppressing the Integrated Sachs-Wolfe effect \cite{Kable:2021yws}. These avenues seem worthwhile investigating in the future.

\section{Discussion and final remarks}\label{sec:6}

In this work, we have shown that the geometric framework based on integrable vectorial nonmetricity leads to an equivalent class of purely kinetic $k$-essence models. A consistent Lagrangian formulation implementing the integrability condition was put forward, giving rise to non-trivial dynamics for the nonmetricity sector. In the cosmological setting, we found a number of important fixed points using the dynamical systems treatment. Notably, this includes a late-time (stable) de Sitter point and a transitory (saddle-like) matter-dominated region at intermediate times. The section of the phase space containing these points corresponds to the \textit{completely symmetric geometry} of Fig.~\ref{fig:vectorial_nonmetricity}, which relies on the existence of our new cubic nonmetricity vector term, see Eq.~(\ref{vectorial_nonmetricity}).
The dynamical systems results also qualitatively confirm that the quadratic kinetic $k$-essence model gives a viable description of the background evolution of the Universe. This is true even when the $k$-essence stability constraints $c_s^2\geq0$, $\rho_{\phi,X} >  0$ are implemented. In fact, these constraints turned out to be vital for the subsequent observational analysis.

The viability of the geometric $k$-essence model was then confirmed with a combination of CC, SNe Ia and BAO data. 
Implementing priors derived from the stability constraints allowed us to avoid unstable regions of the parameter space, which would make the model difficult to constrain.
Hence, this may explain why it has not received much attention in the context of observational analyses. Nonetheless, our results support the claim that the model is comparable to $\Lambda$CDM at late times. This is also evident from the optimised parameters $\{H_0,r_d,\mathcal{M}\}$ obtained from the MCMC analysis, which match their $\Lambda$CDM counterparts, see Table \ref{optimalparameters} and Figs. \ref{contourplots} and \ref{fig:hub}. On the other hand, the matter density parameter $\Omega_{m0}$ takes slightly lower values in our model, due to the intermediate dark matter properties of the $k$-essence fluid. The goodness of fit with the observational data is comparable between both models, with a similar $\chi^2$ in both models, but a preference for $\Lambda$CDM from AIC and BIC, as reported in Table \ref{statisticalresults}. These results strongly motivate further study at the level of linear cosmological perturbations for our specific model.
In this case, a broader sample of cosmological  data could be used to better assess the viability of the ghost condensate $k$-essence model.

Besides further observational investigations, there are many other theoretical routes that could be explored utilising the generalised vectorial nonmetricity. The most straightforward modification is the inclusion of torsion, thereby generalising the vectorial distortion theories discussed in \cite{BeltranJimenez:2015pnp}. It is worth investigating whether, and in what way, torsion modifies the cosmological dynamics. For simplicity, in this work we chose the linear action constructed solely from the Ricci scalar. On the other hand, non-linear actions would instead lead to higher-order equations of motion and non-minimal couplings between the scalar kinetic term 
% $\overset{\circ}{\nabla}_{\mu} \phi$ the overset \circ just messes up the text height in paragraphs. I now think 'kinetic term' is clear enough on its own
and the Riemannian sector. Such couplings would drastically change the phenomenology in cosmology, though potentially introducing new instabilities and pathologies. % reword this: 
In a more general EFT approach, recent work has demonstrated that the completely symmetric nonmetricity can in fact provide a ghost-free and radiatively stable framework \cite{Barker:2025xzd}, which aligns nicely with our cosmological results.
Another interesting option would be to investigate the hypermomentum currents within our framework. In the general metric-affine setting, such currents are well-motivated theoretically and have important effects in cosmology \cite{Hehl:1994ue}. In Table \ref{table:dilationshear} we showed that the shear and dilation currents are non-trivial in our geometries, and so we would also expect similar implications for our model. 

Finally, it is worth again emphasising the geometric foundations of the $k$-essence model studied in this work. Using our theoretically consistent Lagrangian formalism, the non-trivial cosmological dynamics result entirely from the Ricci scalar, with no additional terms needing to be added by hand. This makes the framework especially appealing, with a natural geometric origin. Additionally, we have firmly shown that the late-time behaviour is consistent with cosmological observations. It is then clear that the confrontation of this particular model with early-time data is the next crucial step. This will determine if geometric $k$-essence is truly a viable candidate to describe our Universe.

\appendix
\section{Coordinate-free treatment of vectorial nonmetricity} \label{appendix_coofree}
In this appendix, we provide a coordinate-free description of connections with vectorial nonmetricity. We also describe in a coordinate-free manner the conditions, which ensure length and volume preservation under (auto)parallel transport.
\begin{definition}
  On a pseudo-Riemannian manifold $(M,g)$ a connection $\nabla \in \text{Conn}(M)$ is called \textbf{a connection with vectorial nonmetricity} if there exists a one-form $\pi \in \Gamma\left( T^{*}M\right)$, such that
    \begin{equation}\label{coordinatefree}
        \nabla_{X} Y=\overset{\circ}{\nabla}_{X} Y +\frac{2 c_2 -c_1}{2} g(X,Y) P+ \frac{c_1}{2} \left( \pi(Y) X + \pi(X) Y \right) + c_3 \pi(X) \pi(Y) P \, ,
    \end{equation}
    where $\overset{\circ}{\nabla}$ denotes the Levi-Civita connection of $g$, $X,Y \in \Gamma(TM)$ are vector fields and $P$ is the dual vector field associated to $\pi$, i.e. $g(X,P)=\pi(X) \, .$
\end{definition}
\begin{remark}
    In a local coordinate system, the Christoffel symbols of a connection with vectorial nonmetricity are given by
    \begin{equation}
        \tensor{\Gamma}{^\lambda _\mu _\nu}=\overset{\circ}{\Gamma} \tensor{}{^\lambda _\mu _\nu} + \frac{2 c_2 - c_1}{2} g_{\mu \nu} \pi^{\lambda} + \frac{c_1}{2}\left( \delta{^\lambda}_\mu \pi_\nu +\delta^\lambda_\nu \pi_\mu \right) + c_3 \pi_\mu \pi_\nu \pi^\lambda \, .
    \end{equation}
\end{remark}
\begin{proof}
    We choose a local chart, in which $X=\partial_{\nu}, Y=\partial_{\mu}$, and the coordinates of the dual vector field are given by $P^{\lambda}=\pi^{\lambda}$.
    Plugging in the local coordinate expressions for $X$ and $Y$ in \eqref{coordinatefree} leads to
    \begin{equation}
        \nabla_{\partial_\nu} \partial_{\mu}=\overset{\circ}{\nabla}_{\partial_\nu} \partial_{\mu} + \frac{2c_2-c_1}{2} g(\partial_{\nu},\partial_{\mu}) \pi^\rho \partial_\rho+\frac{c_1}{2} \left(\pi(\partial_\mu) \partial_{\nu}+ \pi(\partial_\nu) \partial_\mu \right)+c_3 \pi(\partial_\nu) \pi(\partial_\mu) P^\rho \partial_{\rho} \, .
    \end{equation}
    To extract the Christoffel symbols from this expression, we act with $dx^{\lambda}$ and directly obtain the desired result
    \begin{equation}
        \tensor{\Gamma}{^\lambda _\mu _\nu}=\overset{\circ}{\Gamma}\tensor{}{^\lambda _\mu _\nu}+ \frac{2c_2 -c_1}{2} g_{\mu \nu} \pi^\lambda+\frac{c_1}{2} \left(\pi_\mu \delta^{\lambda}_{\nu} + \pi_{\nu} \delta^{\lambda}_{\mu} \right) + c_3 \pi_\nu \pi_\mu \pi^{\lambda} \, .
    \end{equation}
    \end{proof}
    \begin{remark}
        It is well known that an affine connection is uniquely characterized by its nonmetricity $Q(X,Y,Z)= - \nabla_X g(Y,Z)$ and torsion $T=\nabla_X Y-\nabla_Y X - [X,Y]$. Since a connection with vectorial nonmetricity is by definition symmetric, it follows that it is torsion-free. Its nonmetricity tensor is given by
        \begin{equation}
            Q(X,Y,Z)=c_1 \pi(X) g(Y,Z)+c_2\left( \pi(Z) g(X,Y) + \pi(Y) g(X,Z) \right) + 2c_3 \pi(X) \pi(Y) \pi(Z) \, .
        \end{equation}
    \end{remark}
    In \cite{Csillag:2024eor} it was shown that a special class of affine connections preserves the length of vectors during autoparallel transport. Here, we show how this property constrains the values of the coefficients $c_1$, $c_2$, and $c_3$ for connections with vectorial nonmetricity.
    \begin{proposition}
       Let $\nabla$ be a connection with vectorial nonmetricity. Then, for all   vector fields $X \in \Gamma(TM)$, which  satisfy $\nabla_{X} X=0$, the following are equivalent:
       \begin{enumerate}
          \item[$(i)$] $\nabla$ preserves lengths under parallel transport, i.e. $\nabla_{X}(g(X,X))=0\,.$
          \item[$(ii)$] The coefficients of the vectorial nonmetricity satisfy $c_1=-2c_2$, $c_3=0 \, .$ 
       \end{enumerate}
    \end{proposition}
    \begin{proof}
    For any vector field $X \in \Gamma(TM)$ satisfying  $\nabla_{X}X=0$, substituting $X=Y$ in \eqref{coordinatefree} leads to
    \begin{equation}\label{expressfromhere}
        \nabla_{X} X=\overset{\circ}{\nabla}_{X} X +\frac{2c_2-c_1}{2} g(X,X) P + \frac{c_1}{2} (\pi(X) X + \pi(X) X) +c_3 \pi(X) \pi(X) P \, .
    \end{equation}
    Since $g(X,X)$ is a smooth function, both covariant derivatives act on it the same way, so we have
    \begin{equation}
        \nabla_{X}(g(X,X))=\overset{\circ}{\nabla}_{X}(g(X,X))=\big(\overset{\circ}{\nabla}_{X}g\big)(X,X) + g\big( \overset{\circ}{\nabla}_{X} X,X \big)+ g\big(X, \overset{\circ}{\nabla}_{X} X \big) \, .
    \end{equation}
    Using that the Levi-Civita connection is metric-compatible, this yields
    \begin{equation}\label{pluginhere}
        \nabla_{X}(g(X,X))=2g \big(\overset{\circ}{\nabla}_{X} X,X \big) \, .
    \end{equation}
    Note that we can express $\overset{\circ}{\nabla}_{X} X$ from \eqref{expressfromhere} as
    \begin{equation}
    \overset{\circ}{\nabla}_{X} X=\nabla_{X} X -\frac{2c_2-c_1}{2} g(X,X) P - c_1 \pi(X) X - c_3 \pi(X) \pi(X) P \, .
    \end{equation}
    Plugging this back into \eqref{pluginhere}, we obtain
    \begin{equation}
        \nabla_{X}\left( g(X,X) \right)=2g \Big(\nabla_{X} X -\frac{2c_2-c_1}{2} g(X,X) P - c_1 \pi(X) X - c_3 \pi(X) \pi(X) P ,X\Big) \, .
    \end{equation}
    Hence, $\nabla_{X}(g(X,X))=0$ is equivalent to
    \begin{equation}
        2g \Big(\nabla_{X} X -\frac{2c_2-c_1}{2} g(X,X) P - c_1 \pi(X) X - c_3 \pi(X) \pi(X) P ,X\Big)=0 \, .
    \end{equation}
    Using $C^{\infty}(M)$ linearity of the metric, the equation takes the equivalent form
    \begin{equation}
        {2g \left(\nabla_{X}X,X \right)}-(2c_2-c_1) g(X,X) g(X,P) - 2c_1 \pi(X) g(X,X) - 2 c_3 \pi(X) \pi(X) g(P,X)=0 \, .
    \end{equation}
   The first term vanishes, due to the first assumption of the theorem. Finally, using the definition of the dual vector field $g(X,P)=\pi(X)$, we are left with
    \begin{equation}
        -(2c_2-c_1) g(X,X) \pi(X)-2 c_1 \pi(X) g(X,X)- 2c_3 \pi(X) \pi(X) \pi(X)=0 \, .
    \end{equation}
    Rearranging, we get
    \begin{equation}
        -(2c_2+c_1) \pi(X) g(X,X) - 2c_3 \pi(X) \pi(X) \pi(X)=0 \, .
    \end{equation}
    Since the above equation has to be satisfied for all vector fields $X$, it is equivalent to
    \begin{equation}
    (2c_2+c_1)=0, \; \; \text{and} \; \; c_3=0 \, , 
    \end{equation}
    which finishes the proof.
    \end{proof}
    \begin{remark}
        This choice of coefficients corresponds to  Schrödinger geometry, as can be seen from Fig.~\ref{fig:vectorial_nonmetricity}.
    \end{remark}
    \begin{proposition}
        For a connection $\nabla$ with vectorial nonmetricity, the following are equivalent:
       \begin{enumerate}
           \item[$(i)$] $\nabla$ preserves angles under parallel transport, that is,
           \begin{equation}\label{anglepreserv}
               \frac{\left|g(X,Y)\right|}{\sqrt{\left|g(X,X)\right|} \sqrt{\left| g(Y,Y)\right|}}=\frac{ \left| \left(\nabla_{Z} g \right)(X,Y) \right|}{\sqrt{\left| \left(\nabla_{Z} g\right)(X,X)\right|} \sqrt{\left| \left(\nabla_{Z} g\right)(Y,Y)\right|}} \, 
           \end{equation}
           for all vector fields $X,Y,Z \in \Gamma(TM)$.
           \item[$(ii)$] The coefficients of the vectorial nonmetricity satisfy $c_2=c_3=0\, .$
       \end{enumerate}
    \end{proposition}
    \begin{proof}
        "$\Rightarrow$": Suppose that $\nabla$ preserves angles, which implies that there exists a one-form $\lambda \in \Gamma(T^*M)$, such that
        \begin{equation}\label{thismusthold}
            (\nabla_{Z} g)(X,Y)=\lambda(Z) g(X,Y) \, .
        \end{equation}
        For a connection with vectorial nonmetricity, the nonmetricity tensor is given by
        \begin{equation}
            (\nabla_{Z} g)(X,Y)=-c_1 \pi(Z) g(X,Y) - c_2 \left( \pi(Y) g(Z,X) + \pi(X) g(Z,Y) \right) - 2 c_3 \pi(Z) \pi(X) \pi(Y) \, .
        \end{equation}
        Hence, for Eq.~\eqref{thismusthold} to hold for all $X,Y,Z \in \Gamma(TM)$, we have to impose $c_2=c_3=0$. In this case, we can identify
        \begin{equation}
            \lambda(Z)=-c_1 \pi(Z) \, .
        \end{equation}
        
        "$\Leftarrow$": Suppose that $c_2=c_3=0$. Then the nonmetricity tensor simplifies to
        \begin{equation}
            (\nabla_{Z}g)(X,Y)=-Q(Z,X,Y)=-c_1 \pi(Z) g(X,Y) \, .
        \end{equation}
        Substituting this into the right hand side of Eq.~\eqref{anglepreserv}, we find
        \begin{equation}
        \begin{aligned}
        \frac{\left|-c_1\pi(Z) g(X,Y)\right|}{\sqrt{\left| - c_1 \pi(Z) g(X,X) \right|} \sqrt{\left| -c_1 \pi(Z) g(Y,Y) \right|}}&=\frac{|c_1| \cdot | \pi(Z)| \cdot |g(X,Y)|}{|c_1| \cdot  | \pi(Z)| \cdot \sqrt{|g(X,X)|} \sqrt{|g(Y,Y)|}}\\
        &=\frac{|g(X,Y)|}{\sqrt{|g(X,X)|} \sqrt{|g(Y,Y)|}} \, ,
        \end{aligned}
        \end{equation}
        which matches the left hand side of Eq.~\eqref{anglepreserv}.
    \end{proof}
    \begin{remark}
        This choice of coefficients corresponds to Weyl geometry, as illustrated in Fig.~\ref{fig:vectorial_nonmetricity}.
    \end{remark}
    It is not widely known in the modified gravity community\footnote{It is worth noting, however, that biconnection gravity exhibits the mathematical structure of a statistical manifold in the presence of an appropriate hypermomentum tensor \cite{PhysRevD.108.044026}.}, but if we set $c_1=c_2$, we obtain a structure known in the mathematical literature as a \textit{statistical manifold} \cite{AmariNagaoka2000,Amari1997,Matsuzoe2010,Lauritzen1987, Iosifidis:2023wbx}.
    \begin{definition}
        A torsion-free connection $\nabla$ on a pseudo-Riemannian manifold $(M,g)$ is called a \textbf{statistical connection}, if it satisfies the Codazzi equation
        \begin{equation}
           \left(\nabla_{X} g \right)(Y,Z)=\left(\nabla_{Y} g \right)(X,Z) \; \; \forall X,Y,Z \in \Gamma(TM) \, .
        \end{equation}
        In this case, the set $(M,g,\nabla)$ is called a \textbf{statistical manifold}.
    \end{definition}
    Note that this really just says something about the nonmetricity, since this immediately implies
    \begin{equation}\label{requirement}
        Q(X,Y,Z)=Q(Y,X,Z) \, .
    \end{equation}
    However, since the nonmetricity tensor is symmetric in its last two indices by definition, the requirement given by Eq.~\eqref{requirement} immediately implies that it is completely symmetric, i.e.
    \begin{equation}
        Q(X,Y,Z)=Q(Y,X,Z)=Q(X,Z,Y)=Q(Z,X,Y) \, .
    \end{equation}
  \begin{proposition}
        A connection $\nabla$ with vectorial nonmetricity on a pseudo-Riemannian manifold $(M,g)$ makes the pair $(M,g,\nabla)$ a statistical manifold iff $c_1=c_2 \, .$
   \end{proposition}
   \begin{proof}
        By direct computation, we have
        \begin{equation}
            Q(X,Y,Z)=c_1 \pi(X) g(Y,Z) + c_2( \pi(Z) g(X,Y)+\pi(Y) g(X,Z))+2c_3 \pi(X) \pi(Y) \pi(Z) \, ,
        \end{equation}
        \begin{equation}
            Q(Y,X,Z)=c_1 \pi(Y) g(X,Z)+c_2( \pi(Z) g(Y,X) + \pi(X) g(Y,Z))+2c_3 \pi(Y) \pi(X) \pi(Z) \, .
        \end{equation}
        Since the left hand sides are equal, it follows that the equality
        \begin{equation}
        \begin{aligned}
           &{c_1 \pi(X) g(Y,Z)} + c_2( {\pi(Z) g(X,Y)}+{\pi(Y) g(X,Z)})+{2c_3 \pi(X) \pi(Y) \pi(Z)}={c_1 \pi(Y) g(X,Z)}\\
           &+c_2( {\pi(Z) g(Y,X)}
            + {\pi(X) g(Y,Z)})
            +{2c_3 \pi(Y) \pi(X) \pi(Z)} 
        \end{aligned}
        \end{equation}        
        must hold for all $X,Y,Z \in \Gamma(TM)$. %The orange and green terms cancel. Regrouping, we get
        Regrouping, we get
          \begin{equation}
         {(c_1-c_2) \pi(X) g(Y,Z)}+{(c_2 - c_1) \pi(Y) g(X,Z)}=0 \; \; \forall X,Y,Z \in \Gamma(TM) \, .
        \end{equation}    
       The latter condition is identically satisfied for all vector fields if and only if $c_1-c_2=c_2-c_1=0$, which implies $c_1=c_2 \, .$
       \end{proof}
    \begin{definition}
        Two statistical manifolds $(M,g,\nabla)$ and $\left(M,\widetilde{g}, \widetilde{\nabla} \right)$ are called projectively Weyl-related if there exist two nonzero smooth functions $f,h \in C^{\infty}(M)$, $f \neq  -h$ such that $\widetilde{g}=e^{f+h} g$ and the following is satisfied
        \begin{equation}
         \widetilde{\nabla}_{X} Y=\nabla_{X} Y+ (df)(X)Y+(df)(Y)X-g(X,Y) (dh)^{\sharp} \, .
        \end{equation}
    \end{definition}
    \begin{proposition}
            For a connection with vectorial nonmetricity $\nabla$  on a pseudo-Riemannian manifold $(M,g)$ the following are equivalent:
            \begin{enumerate}
                \item[$(i)$] $c_1=c_2 \neq 0$, $c_3=0$ and $\pi$ is integrable.
                \item [$(ii)$] There exists a statistical manifold $\Big(M, \widetilde{g}, \overset{\circ}{\widetilde{\nabla}} \Big)$, which is projectively Weyl-related to $(M,g,\nabla) \, .$
            \end{enumerate}
        \end{proposition}
        \begin{proof}
            $"\Rightarrow"$: If $c_1=c_2 \neq 0, c_3=0$ and $\pi$ is integrable, we can write the connection as   \begin{equation}\label{substractfromthis}
                \nabla_{X} Y=\overset{\circ}{\nabla}_{X} Y+\frac{c_1}{2} g(X,Y) (df)^{\sharp}+\frac{c_1}{2} \left((df)(Y) X + (df)(X) Y \right) \, ,
            \end{equation}
            where $\overset{\circ}{\nabla}_{X} Y$ denotes the Levi-Civita connection of $g$ and we have that $\pi=df$ for some smooth function $f \in C^{\infty}(M)$. Since $g$ and $\widetilde{g}:=e^{f}g$ are conformally related\footnote{Note that projectively Weyl-related implies conformally related, but not necessarily vice versa. We abuse notation by denoting both with a tilde.}, their Levi-Civita connections $\overset{\circ}{\nabla}$  and $\overset{\circ}{\widetilde{\nabla}}$ satisfy the identity   \begin{equation}\label{substractthis}
                \overset{\circ}{\widetilde{\nabla}}_{X} Y=\overset{\circ}{\nabla}_{X} Y+\frac{1}{2}\left( (df)(X) Y + (df)(Y) X - g(X,Y) (df)^{\sharp} \right) \, .
            \end{equation}
            By subtracting  \eqref{substractthis} from \eqref{substractfromthis}, we obtain
            \begin{equation}
            \nabla_{X} Y - \overset{\circ}{\widetilde{\nabla}}_{X} Y=\frac{c_1 +1}{2} g(X,Y) \left( df \right)^{\sharp} +\frac{c_1-1}{2} \left((df)(X)Y+(df)(Y) X \right) \, ,
            \end{equation}
            from which it follows that
            \begin{equation}
                \nabla_{X} Y=\overset{\circ}{\widetilde{\nabla}}_{X} Y +\frac{c_1 +1}{2} g(X,Y) \left( df \right)^{\sharp} +\frac{c_1-1}{2} \left((df)(X)Y+(df)(Y) X \right) \, .
            \end{equation}
            Since the derivative is linear, we can absorb the constants and define two functions 
            \begin{equation}
                F:=\frac{c_1+1}{2} f \, , \quad h:=\frac{1-c_1}{2}f \, ,
            \end{equation}
            to obtain
            \begin{equation}
                \overset{\circ}{\widetilde{\nabla}}_{X} Y=\nabla_X Y+ (dh)(X)Y+(dh)(Y) X - g(X,Y) \left( dF \right)^{\sharp} \, .
            \end{equation}
            It can be easily verified that the pair $\left(M, \widetilde{g}, \widetilde{\nabla} \right)$ is a statistical manifold, and that $h \neq -F$. Combining these two statements finishes the proof of the first implication.\\
            "$\Leftarrow$": Supposing there exists a statistical manifold $\Big( M, \widetilde{g}, \overset{\circ}{\widetilde{\nabla}} \Big)$, which is projectively Weyl-related to $(M,g, \nabla)$, we have
            \begin{equation}
                \overset{\circ}{\widetilde{\nabla}}_{X} Y=\nabla_{X} Y + (df)(X) Y + (df)(Y) X - g(X,Y) (dh)^{\sharp}
            \end{equation}
            for some smooth functions $f,h \in C^{\infty}(M)$ that satisfy $f \neq -h$. Then, it immediately follows that
            \begin{equation}
                \nabla_{X} Y=\overset{\circ}{\widetilde{\nabla}}_{X} Y - (df)(X)Y - (df)(Y) X + g(X,Y) (dh)^{\sharp} \, .
            \end{equation}
            Now since $g$ and $\widetilde{g}=e^{f+h}g$ are conformally related, their Levi-Civita connections satisfy
            \begin{equation}
            \begin{aligned}
                \overset{\circ}{\widetilde{\nabla}}_{X} Y=\overset{\circ}{\nabla}_{X} Y&+ \frac{1}{2} \left((df)(X)Y+(dh)(X)Y+(df)(Y)X \right)\\
                &+\frac{1}{2} \left((dh)(Y)(X) - g(X,Y) (df)^{\sharp}-g(X,Y) (dh)^{\sharp} \right) \, .
            \end{aligned}
            \end{equation}
            Consequently, we get
            \begin{equation}
                \nabla_{X} Y=\overset{\circ}{\nabla}_{X} Y- \frac{1}{2} \left(df(X)Y-(dh)(X)Y+(df)(Y)X -(dh)(Y)X +g(X,Y) (df)^{\sharp} - g(X,Y) (dh)^{\sharp} \right) \, ,
            \end{equation}
            which can be equivalently rewritten as
            \begin{equation}
                \nabla_{X} Y=\overset{\circ}{\nabla}_{X} Y -\frac{1}{2} \left(d(f-h)(X) Y+d(f-h)(Y)X+g(X,Y)\left(d(f-h) \right)^{\sharp} \right) \, .
            \end{equation}
            Since $f \neq -h$, for a non-zero $c_1$ we can introduce the notation $F:=-\frac{f-h}{c_1}$ to obtain
            \begin{equation}
                \nabla_{X}Y=\overset{\circ}{\nabla}_{X} Y +\frac{c_1}{2}\left( (dF)(X) Y + (dF)(Y)X + g(X,Y) (dF)^{\sharp} \right) \, ,
            \end{equation}
            which is, as desired, a connection with vectorial nonmetricity satisfying $c_1=c_2, c_3=0$ with integrable $\pi=dF \, .$
        \end{proof}
        \begin{definition}
            Two connections $\nabla, \nabla^{*}$ on a pseudo-Riemannian manifold $(M,g)$ are called dual with respect to $g$ if
            \begin{equation}
                X(g(Y,Z))=g\left(\nabla_{X} Y, Z \right)+g\left(Y,\nabla_{X}^{*} Z \right) \, .
            \end{equation}
        \end{definition}
        \begin{proposition}
            The torsion $T^{*}$ and nonmetricity $Q^{*}$ of the dual connection associated to a connection with vectorial nonmetricity $\nabla$ are given by
            \begin{align}
                Q^{*}(X,Y,Z &)=-c_1 \pi(X) g(Y,Z) - c_2 \left( \pi(Z) g(X,Y) + \pi(Y) g(X,Z) \right) - 2c_3 \pi(X) \pi(Y) \pi(Z) \, , \\
                g(T^{*}(X,Y),Z) &=(c_2-c_1) \left( \pi(X) g(Y,Z) - \pi(Y) g(X,Z) \right) \, .
            \end{align}
        \end{proposition}
        \begin{proof}
            We use the identities of Theorem A.1. as presented in \cite{Csillag_2024statmfd}, which hold for an arbitrary affine connection
            \begin{equation}
                Q^{*}(X,Y,Z)=-Q(X,Y,Z) \, , \quad g(T^{*}(X,Y),Z)=g(T(X,Y),Z)-Q(X,Y,Z)+Q(Y,X,Z) \, .
            \end{equation}
            Applying this to the connection \eqref{coordinatefree}, we get the nonmetricity of its dual connection
             \begin{equation}
                Q^{*}(X,Y,Z)=-c_1 \pi(X) g(Y,Z) - c_2 \left( \pi(Z) g(X,Y) + \pi(Y) g(X,Z) \right) - 2c_3 \pi(X) \pi(Y) \pi(Z) \, .
            \end{equation}
           The torsion of the dual connection is obtained through the computation
            \begin{equation}
            \begin{aligned}
                g(T^{*}(X,Y),Z)= &{-c_1 \pi(X) g(Y,Z)} - c_2 \left( {\pi(Z) g(X,Y)} +{\pi(Y) g(X,Z)} \right) {-2c_3 \pi(X) \pi(Y) \pi(Z)}\\
                &{+c_1\pi(Y) g(X,Z)} + c_2 \left( {\pi(Z) g(Y,X)} {+ \pi(X) g(Y,Z)} \right) {+2c_3 \pi(Y) \pi(X) \pi(Z)}\\
                =&\left(c_2 -c_1 \right)\left(\pi(X) g(Y,Z) - \pi(Y) g(X,Z) \right) \, .
            \end{aligned}
            \end{equation}
        \end{proof}
        \begin{corollary}
            If $c_1 \neq c_2$, then the dual connection has a semi-symmetric type of torsion \cite{Csillag:2024oqo, universe10110419}, and consequently the triple $\left(M,g,\nabla^{*}\right)$ becomes a quasi-statistical manifold, even if $(M,g,\nabla)$ is neither a statistical manifold, nor a quasi-statistical manifold \cite{Csillag_2024statmfd}. Hence, taking the dual of a connection with vectorial nonmetricity gives a novel example of a semi-symmetric nonmetric connection, which recovers as special cases:
            \begin{itemize}
                \item[$\triangleright$] the semi-symmetric Weyl connection for the choice $c_1=1$, $c_2=c_3=0$ \cite{Csillag_2024statmfd},
                \item[$\triangleright$] the semi-symmetric Schrödinger connection for the choice $c_1=1$, $c_2=-\frac{1}{2}$, $c_3=0$ \cite{Csillag_2024statmfd},
                \item[$\triangleright$] the semi-symmetric nonmetric connection for the choice $c_1=0$, $c_2=1$, $c_3=0$ \cite{AYDIN2025128795}.
            \end{itemize}
        \end{corollary}
        \begin{remark}
            From the previous proposition, we see that there is a natural information geometric procedure to generate torsion for a connection with vectorial nonmetricity. This construction generalises geometries with linear distortion studied in \cite{BeltranJimenez:2015pnp}, by allowing for a cubic term in $\pi$. The effects of torsion in the present framework could be studied in a follow-up.
        \end{remark}
        \begin{proposition}
            For a connection with integrable vectorial nonmetricity satisfying $c_1=c_2$ the following are equivalent
            \begin{enumerate}
                \item[$(i)$] There exists a volume form $\omega$ which is covariantly conserved, i.e. $\nabla \omega=0 \, .$
                \item[$(ii)$] $
    \overset{\circ}{\nabla}_{\nu} \left( 3 c_1 \overset{\circ}{\nabla}_{\mu} \phi + c_3 \overset{\circ}{\nabla}_{\mu} \phi \overset{\circ}{\nabla}_{\rho} \phi \overset{\circ}{\nabla}\tensor{}{^\rho} \phi \right)=\overset{\circ}{\nabla}_{\mu} \left( 3c_1 \overset{\circ}{\nabla}_{\nu} \phi + c_3 \overset{\circ}{\nabla}_{\nu} \phi \overset{\circ}{\nabla}_{\rho} \phi \overset{\circ}{\nabla} \tensor{}{^\rho}\phi \right) \, .$    
            \end{enumerate}
        \end{proposition}
\begin{proof}
     For the choice $c_1=c_2$, the pair $(M,g,\nabla)$ is a statistical manifold. Hence, we use Proposition 3.2. of \cite{MATSUZOE2006567}, which states that the connection $\nabla$ is equiaffine  iff the one-form\footnote{This one-form is called the Chebyshev form in the case of statistical manifolds.}
     \begin{equation}
         Q_{\sigma}=Q_{\sigma \alpha \rho} g^{\alpha \rho}
     \end{equation}
     is closed, $\overset{\circ}{\nabla}_{\mu} Q_{\nu}= \overset{\circ}{\nabla}_{\nu} Q_{\mu}$ . For this trace, one immediately obtains
     \begin{equation}
          Q_{\sigma}=6c_1 \overset{\circ}{\nabla}_{\sigma} \phi + 2c_3 \overset{\circ}{\nabla}_{\sigma} \phi \overset{\circ}{\nabla}_{\rho} \phi \overset{\circ}{\nabla}\tensor{}{^\rho} \phi \, .
     \end{equation}
     Hence, $\nabla$ being equiaffine is equivalent to\footnote{Multiplying a one-form by a non-zero constant does not change the fact that it is closed.}
    \begin{equation}\label{mycond}
        \overset{\circ}{\nabla}_{\nu} \left( 3 c_1 \overset{\circ}{\nabla}_{\mu} \phi + c_3 \overset{\circ}{\nabla}_{\mu} \phi \overset{\circ}{\nabla}_{\rho} \phi \overset{\circ}{\nabla}\tensor{}{^\rho} \phi \right)=\overset{\circ}{\nabla}_{\mu} \left( 3c_1 \overset{\circ}{\nabla}_{\nu} \phi + c_3 \overset{\circ}{\nabla}_{\nu} \phi \overset{\circ}{\nabla}_{\rho} \phi \overset{\circ}{\nabla} \tensor{}{^\rho}\phi \right) \, ,
    \end{equation}
    which, by the definition of an equiaffine connection, is equivalent to $(i)$.
\end{proof}
\begin{remark}
    An alternative and more direct proof uses Proposition 2.1. of \cite{MATSUZOE2006567}. For $c_1=c_2$ the antisymmetric part of the Ricci tensor of a connection with integrable vectorial nonmetricity is given by
    \begin{equation}\label{Erikcondition}
    R_{[\mu \nu]} = 2 c_3  \overset{\circ}{\nabla}{^\lambda}  \phi \left( \overset{\circ}{\nabla}_{\mu} \phi (\overset{\circ}{\nabla}_{\lambda} \overset{\circ}{\nabla}_{\nu} \phi) -  \overset{\circ}{\nabla}_{\nu} \phi (\overset{\circ}{\nabla}_{\lambda} \overset{\circ}{\nabla}_{\mu} \phi) \right) \, .
\end{equation}
Imposing this antisymmetric part to vanish (c.f. Proposition 2.1. of \cite{MATSUZOE2006567}) is equivalent to Eq.~\eqref{mycond}.
\end{remark}

\section{Derivation of the field equations} \label{append_var}
We restrict our attention to the gravitational sector, as the variation of the matter action is straightforward due to the absence of any non-minimal couplings between geometry and matter. The full action then reads
\begin{equation}
    S[g,\pi,\phi,\lambda]=\frac{1}{2 \kappa} \int \sqrt{-g} \Big(R+ \xi \nabla_\mu \pi^\mu +\lambda^\mu \big( \pi_\mu - \overset{\circ}{\nabla}{}_{\mu} \phi \big) \Big) d^4x \, ,
\end{equation}
and it can be equivalently rewritten as
\begin{equation}
    \frac{1}{2 \kappa} \int \sqrt{-g} \Big(\overset{\circ}{R}+b_3 \overset{\circ}{\nabla}_{\mu} \pi^{\mu} + b_1 \pi_\mu \pi^\mu+\frac{1}{2} b_2 \pi_\rho \pi^\rho \pi_\lambda \pi^\lambda +\lambda^\mu \big(\pi_\mu - \overset{\circ}{\nabla}_{\mu} \phi \big) \Big)  d^4x  \, .
\end{equation}
Note that the $b_3$ is a boundary term and will not contribute to the equations of motion. Varying the remaining terms with respect to $g_{\mu \nu}$, $\pi_{\mu}$, $\lambda^{\mu}$ and $\phi$ yield
\begin{equation}\label{ActionVar}
\begin{aligned}
    \overset{\circ}{R}_{\mu \nu} - \frac{1}{2} g_{\mu \nu} \overset{\circ}{R}+  b_1 \pi_\mu \pi_\nu - \frac{1}{2} g_{\mu \nu} b_1 \pi^\rho \pi_\rho+b_2 \pi_\mu \pi_\nu \pi_\rho \pi^\rho -\frac{1}{4}b_2 g_{\mu \nu} \pi_\rho \pi^\rho \pi_\lambda \pi^\lambda -\frac{1}{2} g_{\mu \nu} \lambda^\rho \Big( \pi_\rho - \overset{\circ}{\nabla}_{\rho} \phi \Big)&=0 \, ,\\
    b_1 \pi^\mu + b_2 \pi_\rho \pi^\rho \pi^\mu +{\frac{1}{2}} \lambda^\mu&=0 \, ,\\
    \pi_{\mu}-\overset{\circ}{\nabla}_{\mu} \phi&=0\, ,\\
    \overset{\circ}{\nabla}_{\mu} \lambda^\mu&=0 \, .
\end{aligned}
\end{equation}
From the variation with respect to the Lagrange multiplier we obtain the integrable condition $\pi_\mu=\overset{\circ}{\nabla}_{\mu} \phi$. Plugging this back into the metric field equation gives
\begin{equation}
     \overset{\circ}{G}_{\mu \nu} + \overset{\circ}{\nabla}_{\mu} \phi \overset{\circ}{\nabla}_{\nu} \phi \Big( b_1 + b_2 \overset{\circ}{\nabla}_{\rho} \phi \overset{\circ}{\nabla} \tensor{}{^\rho} \phi \Big)-\frac{1}{2} g_{\mu \nu} \overset{\circ}{\nabla}_{\rho} \phi \overset{\circ}{\nabla} \tensor{}{^\rho} \phi \Big( b_1 +\frac{b_2}{2} \overset{\circ}{\nabla}_{\lambda} \phi \overset{\circ}{\nabla}{^\lambda} \phi \Big)=0 \, ,
\end{equation}
which is equivalent to \eqref{EoM_met}. To obtain the dynamical equation for the  scalar field, we take the covariant divergence of the second equation appearing in \eqref{ActionVar} to arrive at
\begin{equation}
    b_1 \overset{\circ}{\nabla}_{\mu} \pi^{\mu}+ b_{2}  \overset{\circ}{\nabla}_{\mu} (\pi_\rho \pi^\rho \pi^\mu)+\overset{\circ}{\nabla}_{\mu} \lambda^{\mu}=0 \, .
\end{equation}
Using the variation with respect to $\phi$, the last term vanishes. Finally, using the integrable condition, the scalar field equation reduces to
\begin{equation}
      b_1 \overset{\circ}{\nabla}_{\mu} \overset{\circ}{\nabla} \tensor{}{^{\mu}} \phi + b_2 \Big( \overset{\circ}{\nabla}_{\mu} \phi \overset{\circ}{\nabla}\tensor{}{^{\mu}}\phi \overset{\circ}{\nabla}_{\nu} \overset{\circ}{\nabla}\tensor{}{^{\nu}} \phi + 2  \overset{\circ}{\nabla}_{\mu} \phi\overset{\circ}{\nabla}\tensor{}{^{\nu}} \phi \overset{\circ}{\nabla}_{\nu} \overset{\circ}{\nabla}\tensor{}{^{\mu}} \phi \Big) = 0  \, ,
\end{equation}
which is in agreement with \eqref{EoM_phi}.

%%%%%%%%%%%%%%%%%%%%%%%%%%%%%%%%%%%%

\section{Comment on alternative dynamical systems approaches} 
\label{sec:comment}
Given the revitalised interest in the quadratic power-law $k$-essence models \cite{Quiros:2025fnn,Hussain:2024qrd}, it is worthwhile to briefly comment on the differing conclusions appearing in the literature and offer potential explanations.
Firstly, our results are consistent with the dynamical systems analysis performed in \cite{Fang:2014qga,Chakraborty:2019swx}. Regarding the more recent work \cite{Quiros:2025fnn}, our initial analysis produces the same results, with our $k$-essence stability constraints leading to the same restrictions on the physical phase space. That work is also the first to use these perturbative stability constraints in conjunction with the dynamical systems analysis, a tactic which we readily adopt in this paper.
However, the alternative set of variables lead to a markedly different presentation. We will refer to these variables used in \cite{Quiros:2025fnn} as 
\begin{equation*}
    \widetilde{x}=\frac{b_1 X}{b_1 X+H^2} \, , \quad   \widetilde{y}= \frac{b_2 H^2}{b_1^2 + b_2 H^2} \, ,
\end{equation*}
where $b_1$ and $b_2$ are assumed to be positive.

These alternative variables $\left(\widetilde{x},\widetilde{y}\right)$ have the benefit of being manifestly compact. However, after implementing the stability constraints $c_s^2 \geq 0$ and $0 \leq \Omega_{m} \leq 1$, the physical phase space in our original variables $(x,y)$ is also compact, so this advantage is no longer relevant. The problem that plagued our formulation, and the previous works using the same variables \cite{Fang:2014qga,Chakraborty:2019swx}, is that the line $x=-2y$ and the `would-be' critical point O is divergent. As explained in detail in Section \ref{sec:fixed}, this cannot be properly captured by the dynamical systems formulation because the limit is not well-defined.
This problem is still present in the $\left(\widetilde{x},\widetilde{y}\right)$ formulation. In fact, it is somewhat worse: both lines $\widetilde{y}=1$ and $\widetilde{x}=1$ are divergent, and these lines contain three `would-be' critical points (which we will now simply call critical points). 

The first item of disagreement is on the physical properties of the critical points on this $\widetilde{y}=1$ divergent line: the unstable early-time fixed point at $\left(\widetilde{x}=0,\widetilde{y}=1\right)$ does not behave like matter (as claimed in \cite{Quiros:2025fnn}) when approached from any trajectories within the physical phase space.
Instead, it behaves like radiation $q=1$, $w_{\textrm{eff}}=1/3$ and is analogous to our early-time point A of Table \ref{tab:fixed}. This can be confirmed numerically by evolving any initial conditions within the physical phase space backwards in time. Consequently, this then aligns with other dynamical systems literature studying this model \cite{Fang:2014qga,Chakraborty:2019swx} and the initial analysis of Scherrer \cite{Scherrer:2004au}.

The second point to address is the claim that the model (after implementing stability conditions) doesn't allow for matter-dominated behaviour \cite{Quiros:2025fnn}. This is a subtle issue, as we also do not technically find matter dominated points either. Nonetheless, we know that the effective fluid behaves like matter-domination in the vicinity of the origin, see Fig.~\ref{fig:A}. Moreover, when looking only in the physically relevant phase space (Fig.~\ref{fig:prior}), the origin acts like a saddle point by attracting trajectories from point A and projecting them towards point C. Consequently, we are able to find suitable $\Lambda$CDM-like evolutions (Fig.~\ref{fig:evol}) with an extended period of matter domination followed by accelerated expansion. In the alternative formulation with variables $\left(\widetilde{x},\widetilde{y}\right)$, these solutions are equally possible to obtain. Specifically, the `matter-dominated' region is a section of the physical phase space between the early-time radiation repeller and the late-time de Sitter attractor. Numerical solutions to the system $d\widetilde{x}/dN$ and $d\widetilde{y}/dN$ indeed give completely equivalent results once our initial conditions are mapped to the variables used in \cite{Quiros:2025fnn}. This confirms that the purely kinetic $k$-essence models can give rise to suitable cosmologies, and Section \ref{sec:5} proves they are compatible with late-time observational data.

\section{Analytic solution} \label{appendix:analyticalsolution}

Here we present the analytic solution for the cosmological model governed by equations (\ref{ff1})-(\ref{Phi1}). As shown in \cite{Chimento:2003ta,Scherrer:2004au}, the scalar field equation for any purely-kinetic $k$-essence model $P(X)$ can be written in terms of the scale factor $a$ as \cite{Scherrer:2004au}
\begin{equation*}
    (P_{,X} + 2X P_{,XX})a \frac{d X}{d a} + 6 X P_{,X} = 0 \, ,
\end{equation*}
yielding the solution
\begin{equation} \label{Psol}
    X P_{,X}^2 = \frac{K_0}{a^{6}} \, ,
\end{equation}
where $K_0$ is an integration constant. Our model (\ref{Pmodel}) in the reduced variables (\ref{new_vars}) takes the form
\begin{equation}
    P(\widetilde{X})  = H_0^2 \left(- \widetilde{X} + \frac{4}{3}B \widetilde{X}^2 \right) \, 
\end{equation}
where we have defined the dimensionless kinetic term $\widetilde{X}=b_1 X/H_0^2$.
The solutions to (\ref{Psol}) are then the roots of the cubic equation
\begin{equation}
    \widetilde{X} - \frac{16}{3} B \widetilde{X}^2 + \frac{64}{9} B^2 \widetilde{X}^3 - k_0 (1+z)^6 =0 \, ,
\end{equation}
where we have used that $a = 1/(1+z)$ and defined $k_0 = K_0/H_0^2$. Given that both $B$ and $\widetilde{X}$ must be positive, we also require $k_0 \geq 0$ for real solutions. Only one branch of solution is well-defined for all $z>-1$, which can be written as 
\begin{equation}
    \widetilde{X}(z) = \frac{1}{4B} + \frac{1+c_z^2}{8 B c_z} \, , \quad  \textrm{with} \quad c_z := \Big(-1 + 6 \sqrt{2} \sqrt{B k_0(18 B k_0 (1+z)^6 -1)} (1+z)^3 + 36 B k_0 (1+z)^6\Big)^{\frac{1}{3}} \, .
\end{equation}
Note that for $k_0=0$ the kinetic term becomes a constant $\widetilde{X}=3/(8B)$ and the model reduces exactly to $\Lambda$CDM.
A closed-form analytic expression for $h(z)$ can then be obtained by plugging this solution into the Friedmann equation
\begin{equation} \label{hfull}
    h(z)^2 = \Omega_{m0} (1+z)^3 - \frac{\widetilde{X}(z)}{3} + \frac{4}{3} B \widetilde{X}(z)^2 \, .
\end{equation}
Evaluating the Friedmann equation at $z=0$ reveals that each $B$ is associated with a maximum allowed $\Omega_{m0}$, which directly corresponds to the $\Lambda$CDM limit $k_0=0$. These are the constraints derived in (\ref{model_constr1}). Note that the $z=0$ Friedmann constraint is too complicated to explicitly eliminate $k_0$ in terms of $\Omega_{m0}$, which is why we used numerical solutions in Section \ref{sec:5}. However, we have verified that after fixing $\Omega_{m0}$ and $k_0$, the numerical and analytic solutions are equivalent.

Expanding (\ref{hfull}) around future infinity $z \rightarrow -1$ reveals the following form of the solution
\begin{equation} \label{happrox}
    h(z)^2 = \frac{1}{16B} + \left(\Omega_{m0}+\sqrt{\frac{k_0}{6B}} \right)(1+z)^3 + \mathcal{O}(1+z)^6 \, .
\end{equation}
The first cosmological constant term is also exactly what we found for the fixed point C in the dynamical systems analysis, see Section \ref{sec:fixed}. The next leading-order contribution of the $k$-essence field behaves like matter $\propto (1+z)^3$, as pointed out by Scherrer \cite{Scherrer:2004au}. However, this is only true in the vicinity of the fixed point at future infinity. Nonetheless, using the Friedmann constraint for the full solution  (\ref{hfull}) at $z=0$ to eliminate the constant $k_0$ in favour of the matter density parameter $\Omega_{m0}$, we see that the evolution of $h(z)$ is primarily dominated by the $k$-essence terms at low $z$.

This is confirmed by the full analytic solutions plotted in Fig.~\ref{fig:Append}, which track the $\Lambda$CDM solution (i.e., $k_0=0$) for small redshifts almost independently of $\Omega_{m0}$. The left figure shows that the solution $h(z)$ is mainly determined by the parameter $B$ for small $z$. Even for these significantly large changes in $\Omega_{m0}$, the $k$-essence solution $h(z)$ only deviates slightly from the equivalent $\Lambda$CDM curve.
On the other hand, for large $z$ (right figure) both $\Omega_{m0}$ and $B$ are important. The early-time behaviour thus shows a radically different evolution from $\Lambda$CDM. This is also clear from the dynamical systems section, with trajectories originating at the radiation-tracking fixed point. It is therefore clear that early-time data would be especially important for constraining $\Omega_{m0}$.

\begin{figure}[htb!] 
    \centering
\includegraphics[width=0.95\linewidth]{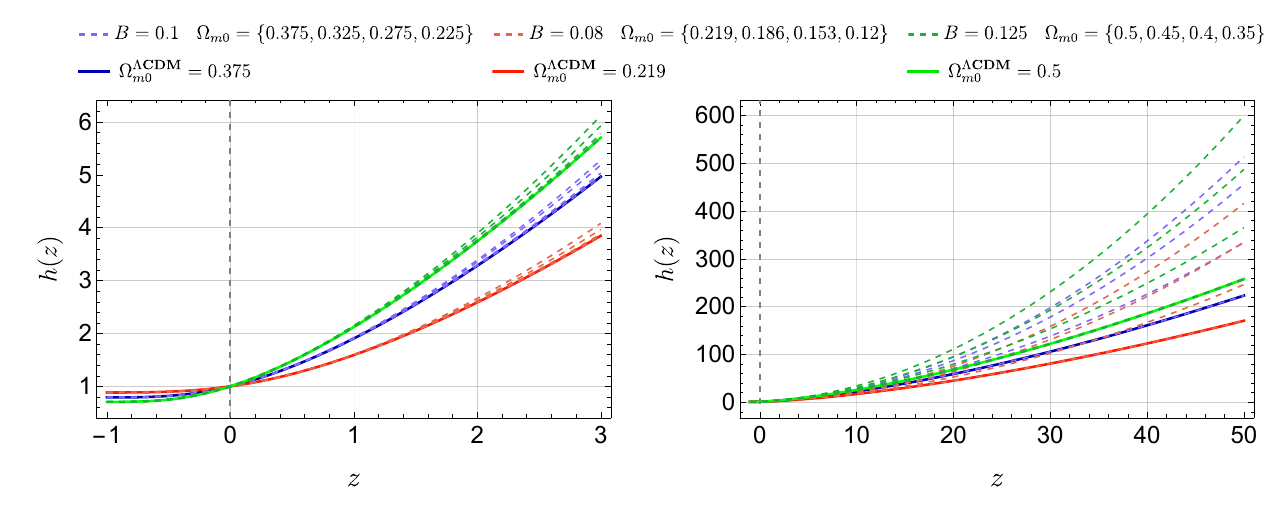}
        \caption{Plots of $h(z)$ for various model parameters of the geometric $k$-essence model (dashed) against $\Lambda$CDM (solid).} \label{fig:Append}
\end{figure}

\acknowledgments{We acknowledge useful discussions with Christian Boehmer, Tiberiu Harko, Andronikos Paliathanasis, Jackson Levi Said, Elsa Teixeira and Sunny Vagnozzi.
This paper is based upon work from COST Action CA21136 Addressing observational tensions in cosmology with systematics and fundamental physics (CosmoVerse) supported by COST (European Cooperation in Science and Technology). E.J. is supported by the Engineering and Physical Sciences Research Council [EP/W524335/1]. The work of L.Cs. was supported by a grant of the Ministry of Research, Innovation and Digitization, CNCS/CCCDI-UEFISCDI, project number PN-IV-P8-8.1-PRE-HE-ORG-2023-0118, within PNCDI IV and Collegium Talentum Hungary.
}

% Bibliography

%% [A] Recommended: using JHEP.bst file
 \bibliographystyle{JHEP}
 \bibliography{biblio.bib}

%% or
%% [B] Manual formatting (see below)
%% (i) We suggest to always provide author, title and journal data or doi:
%% in short all the informations that clearly identify a document.
%% (ii) please avoid comments such as "For a review'', "For some examples",
%% "and references therein" or move them in the text. In general, please leave only references in the bibliography and move all
%% accessory text in footnotes.
%% (iii) Also, please have only one work for each \bibitem.

\end{document}